\documentclass[a4paper,abstract=true,headings=normal,11pt]{scrartcl}
\setkomafont{author}{\large}
\usepackage{authblk}

\usepackage[utf8]{inputenc}
\usepackage{latexsym}
\usepackage{amsmath}
\usepackage{amsthm}
\usepackage{amssymb}
\usepackage{mathtools}
\usepackage{mathabx}
\usepackage{graphicx}
\usepackage{ifthen}
\usepackage{calc}
\usepackage{stmaryrd}
\usepackage{enumitem}
\usepackage{tikz}
\usepackage[colorlinks,linkcolor=red!80!black,urlcolor=blue!80!black,citecolor=green!70!black]{hyperref}
\usepackage{BOONDOX-cal}

\usepackage[backend=bibtex,style=numeric]{biblatex}
\renewbibmacro*{doi+eprint+url}{\iftoggle{bbx:url}
  {\iffieldundef{doi}{\usebibmacro{url+urldate}}{}}
  {}\newunit\newblock
  \iftoggle{bbx:eprint}
  {\usebibmacro{eprint}}
  {}\newunit\newblock
  \iftoggle{bbx:doi}
  {\printfield{doi}}
  {}}

\allowdisplaybreaks

\makeatletter

\newsavebox{\tempbox}
\renewcommand{\@makecaption}[2]{
  \vspace{10pt}
  \sbox{\tempbox}{\textbf{#1.} #2}
  \ifthenelse{\lengthtest{\wd\tempbox > \linewidth}}{
    \textbf{#1.} #2\par
  }{
    \begin{center}
      \textbf{#1.} #2
    \end{center}
  }
}

\makeatother

\numberwithin{equation}{section}

\numberwithin{figure}{section}

\newtheoremstyle{mythm}{}{}{\itshape}{}{\bfseries}{.}{.5em}{\thmname{#1}~\thmnumber{#2}\ifthenelse{\equal{\thmnote{#3}}{}}{}{~(\thmnote{#3})}}

\newtheoremstyle{mydefn}{}{}{\upshape}{}{\bfseries}{.}{.5em}{\thmname{#1}~\thmnumber{#2}\ifthenelse{\equal{\thmnote{#3}}{}}{}{~(\thmnote{#3})}}

\newtheoremstyle{myremark}{}{}{\upshape}{}{\itshape}{.}{.5em}{\thmname{#1}~\thmnumber{#2}\ifthenelse{\equal{\thmnote{#3}}{}}{}{~(\thmnote{#3})}}

\theoremstyle{mythm}
\newtheorem{theorem}{Theorem}[section]
\newtheorem{lemma}[theorem]{Lemma}

\newtheorem{corollary}[theorem]{Corollary}

\theoremstyle{mydefn}
\newtheorem{definition}[theorem]{Definition}
\newtheorem{example}[theorem]{Example}

\theoremstyle{myremark}
\newtheorem{remark}[theorem]{Remark}
\theoremstyle{mythm}

\newcommand{\uend}{\hfill$\lrcorner$}

\newcounter{claimcounter}
\newenvironment{claim}[1][]{
  \renewcommand{\proof}{\smallskip\par\noindent\textit{Proof. }}
  \medskip\par\noindent \ifthenelse{\equal{#1}{}}{\setcounter{claimcounter}{0}\refstepcounter{claimcounter}\textit{Claim~\arabic{claimcounter}.}
  }{\ifthenelse{\equal{#1}{resume}}{\refstepcounter{claimcounter}\textit{Claim~\arabic{claimcounter}.}
    }{\textit{Claim~#1.}
    }
  }
}{
  \par\medskip
}

\newlist{caselist}{description}{10}
\setlist[caselist]{font=\itshape\mdseries}

\setenumerate[1]{label=(\arabic*)}
\newlist{eroman}{enumerate}{2}
\setlist[eroman,1]{label=(\roman*)}
\setlist[eroman,2]{label=(\alph*)}
\newlist{ealph}{enumerate}{1}
\setlist[ealph]{label=(\Alph*)}
\newlist{mea}{enumerate}{1}
\setlist[mea]{label=(\alph*)}

\newcounter{nlistcounter}

\usepackage{color}
\definecolor{blau}{RGB}{0,84,159}
\definecolor{hellblau}{RGB}{142,168,229}
\definecolor{petrol}{RGB}{0,97,101}
\definecolor{tuerkis}{RGB}{0,152,161}
\definecolor{gruen}{RGB}{87,171,39}
\definecolor{maigruen}{RGB}{189,205,0}
\definecolor{gelb}{RGB}{255,237,0}
\definecolor{orange}{RGB}{255,128,0}
\definecolor{magenta}{RGB}{227,0,102}
\definecolor{rot}{RGB}{204,7,30}
\definecolor{bordeaux}{RGB}{161,16,53}
\definecolor{violett}{RGB}{97,33,88}
\definecolor{lila}{RGB}{122,111,172}
\definecolor{grey}{gray}{0.7}
\definecolor{mittelblau}{RGB}{0,128,255}

\newcommand{\bigmid}{\mathrel{\big|}}
\newcommand{\Bigmid}{\mathrel{\Big|}}

\renewcommand{\tilde}{\widetilde}
\renewcommand{\hat}{\widehat}

\renewcommand{\vec}[1]{\boldsymbol{#1}}

\DeclareMathOperator{\rg}{rg}
\DeclareMathOperator{\dom}{dom}

\newcommand{\Fraisse}{Fra\"{\i}ss{\'e}}

\renewcommand{\phi}{\varphi}
\renewcommand{\epsilon}{\varepsilon}

\newcommand{\Nat}{{\mathbb N}}
\newcommand{\PNat}{{\mathbb N}_{\geq 1}}
\newcommand{\NN}{\Nat}
\newcommand{\NNpos}{\PNat}

\newcommand{\CA}{{\mathcal A}}
\newcommand{\CB}{{\mathcal B}}
\newcommand{\CC}{{\mathcal C}}

\newcommand{\CG}{{\mathcal G}}

\newcommand{\CN}{{\mathcal N}}

\newcommand{\CS}{{\mathcal S}}

\newcommand{\CX}{{\mathcal X}}

\usepackage[textwidth=2cm]{todonotes}

\usepackage{algorithmic}

\algsetup{linenodelimiter=.}

\usepackage{float}
\newfloat{algorithm}{ht}{alg}
\floatname{algorithm}{Algorithm}

\DeclareMathOperator{\dist}{dist}

\DeclareMathOperator{\tp}{tp}
\DeclareMathOperator{\ltp}{ltp}
\DeclareMathOperator{\itp}{itp}

\DeclareMathOperator{\ar}{ar}
\DeclareMathOperator{\qr}{qr}

\DeclareMathOperator{\rk}{rk}

\newcommand{\free}{\ensuremath\textup{free}}

\newcommand{\logic}[1]{\textsf{\upshape #1}}

\newcommand{\FO}{\logic{FO}}

\newcommand{\set}[1]{\ensuremath{\{#1\}}}
\newcommand{\setc}[2]{\ensuremath{\{#1 \mid #2\}}}
\newcommand{\deff}{\ensuremath{\coloneqq}}

\newcommand{\normalised}[1]{\ensuremath{\tilde{#1}}}

\newcommand{\FOplus}{\ensuremath{\FO^+}}
\newcommand{\normFOplus}{\normalised{\FO}{}^+}
\newcommand{\LocFOplus}{\ensuremath{\logic{loc}\FOplus}}
\newcommand{\normLocFOplus}{\ensuremath{\logic{loc}\normFOplus}}

\newcommand{\IS}{\ensuremath{\logic{IS}}}
\newcommand{\normIS}{\ensuremath{\normalised{\IS}}}

\renewcommand{\mid}{:}
\renewcommand{\bigmid}{\, : \,}
\renewcommand{\Bigmid}{\; : \; }

\newcommand{\myrho}{\delta}
\newcommand{\rankletter}{\ensuremath{\varrho}}
\newcommand{\rHalbe}{\ensuremath{\frac{r}{2}}}
\newcommand{\psibl}{\ensuremath{\lambda}} \newcommand{\vvar}{\ensuremath{x}}

\newcommand{\conn}{c}

\begin{document}
\title{A Rank-Preserving Gaifman Normal Form}
\author[1]{Martin Grohe}
\affil[1]{RWTH Aachen University}
\author[2]{Nicole Schweikardt}
\affil[2]{Humboldt-Universität zu Berlin}
\date{}
\maketitle

\begin{abstract}
We introduce a rank measure for first-order logic and prove a ``rank-preserving'' version of Gaifman's theorem. Compared to earlier ``rank-preserving locality theorems'' (in particular, [Grohe, Kreutzer, Siebertz, JACM 2017]), our theorem is not only much simpler, but also yields formulas in exactly the same normal form as Gaifman's original theorem. 

As an application of this theorem, we give a simplified proof of the
main result of [Grohe, Kreutzer, Siebertz, JACM 2017] that first-order
properties of nowhere dense structures can be decided in almost linear time.
\end{abstract}

\section{Introduction}
\label{sec:intro}

\emph{Locality} is a fundamental property of first-order logic with many applications, mainly in proving expressivity lower bounds (e.g., \cite{Fagin75,Hanf65,KuskeS17,Libkin97,Libkin00}) and algorithmic meta theorems (e.g., \cite{DreierMS23,FrickG01,GroheKS17,SchweikardtSV22,Seese96}).
Locality comes in different forms (cf.~\cite{Libkin97,Libkin04}), but arguably the most important is Gaifman's Theorem \cite{Gaifman82}. We say that a formula of first-order logic is in \emph{Gaifman normal form} if it is a Boolean combination of
local formulas and of basic local sentences, i.e., sentences of the form
\begin{equation}\label{eq:intro-bl}
  \exists x_1\cdots\exists x_m\left(\bigwedge_{1\le i<j\le
      m}\dist(x_i,x_j)>2r\ \wedge\bigwedge_{1\le i\le m}\psibl(x_i)\right),
\end{equation}
where $\dist(x,y)>2r$ is a formula expressing that the distance between $x$ and $y$ is greater than $2r$ and the formula $\psibl(x)$ is $r$-local around $x$ (cf.\  Sections~\ref{sec:prel} and~\ref{sec:Gaifman} for all necessary definitions). Then Gaifman's Theorem just states that every formula of first-order logic is equivalent to a formula in Gaifman normal form.
The currently known proofs of this theorem (cf.,
\cite{EbbinghausF99,Gaifman82,KeislerLotfallah2004})
transform $\phi$ into a formula in Gaifman normal form that does not preserve the
quantifier rank, i.e., the local formulas and the formulas $\psibl(x)$ in the basic local
sentences may have a larger quantifier rank than $\phi$. In fact, we give an example---to the best of our knowledge the first such example---showing that an increase in the quantifier rank is necessary~(cf.\  Example~\ref{ex:1}).

For some applications of Gaifman's Theorem, specifically for proving
algorithmic meta theorems for nowhere dense structures
\cite{DreierMS23,GroheKS17,GroheS18,SchweikardtSV22}, this increase in
quantifier rank is a serious problem. The reason is that Gaifman's
Theorem needs to be applied recursively, and the depth of the
recursion depends on the quantifier rank of the formulas that still
need
to be processed. To resolve this issue, Grohe, Kreutzer, and Siebertz
\cite{GroheKS17} introduced a new, two-parameter rank measure
(``$q$-rank at most $\ell$'') and proved a locality theorem
preserving the rank of these formulas with respect to this rank
measure. Unfortunately, this theorem is far more complicated than
Gaifman's original theorem and quite tedious to apply.
The locality theorem of \cite{GroheKS17} only applies to
formulas with at most one free variable. In \cite{GroheS18}
it was extended to arbitrary formulas, and this
was used in \cite{GroheS18,SchweikardtSV22} 
to obtain algorithmic meta theorems for counting and enumeration
problems, respectively.
This extension was even more complicated and,
unfortunately, it turned out to be flawed:~Charlotte Lenz
 found a subtle error in the proof of
\cite[Theorem~7.1]{GroheS18}. The present paper grew out of
our effort to fix this error, which ultimately led to a much simpler
rank-preserving locality theorem. 

We introduce a new measure for the rank of first-order formulas. It is related to the notion of
``$q$-rank at most $\ell$'' of \cite{GroheKS17,GroheS18}, yet no
longer requires the two parameters $q$ and $\ell$. That is, in our
setting the rank of a formula is just one nonnegative integer. Then we
prove a rank-preserving version of Gaifman's original theorem: every
first-order formula $\phi$ is equivalent to a formula $\phi'$
in Gaifman normal form where the rank of all local formulas in
$\phi'$ and the rank of all local formulas $\lambda(x)$ in the
basic local sentences \eqref{eq:intro-bl} in $\phi'$ is at most
the rank of $\phi$
(cf.\  Theorem~\ref{thm:rankGaifman} for the precise statement).

As an application of this theorem, we give a simplified proof of the
main result of \cite{GroheKS17} that first-order properties of nowhere
dense structures can be decided in almost linear time.
We also provide a refined version of the theorem (Theorem~\ref{thm:forGroheS18}) that fixes the error of \cite[Theorem~7.1]{GroheS18}.

The rest of the paper is organised as follows.
Section~\ref{sec:prel} fixes notation on relational structures,
first-order logic $\FO$ and its syntactic extension by distance atoms
$\FOplus$, and provides a normalisation procedure for these formulas.

Section~\ref{sec:Gaifman} provides the necessary background on
Gaifman's normal form theorem for $\FO$ and proves that this normal
form does not preserve the quantifier rank of formulas.

Section~\ref{sec:FOpluspq} introduces the subset $\FOplus[p,q]$ of
$\FOplus$, provides an Ehrenfeucht-\Fraisse\ game  for this logic, and
proves an Ehrenfeucht-\Fraisse\ theorem which states that
indistinguishability by formulas in $\FOplus[p,q]$ coincides with
Duplicator having a winning strategy for this game (Theorem~\ref{thm:EF}).

Section~\ref{sec:local} introduces the subset $\LocFOplus[p,q]$ of
\emph{syntactically local} formulas in $\FOplus[p,q]$ and proves that
each such formula is $r$-local for a particular number $r=\myrho(p,q)$ that only
depends on $p,q$ (Lemma~\ref{lemma:locality}).

Section~\ref{sec:main} contains our technical main result. It proves a
Gaifman normal form theorem  which enables us to transform every
formula $\phi$ in $\FOplus[p,q]$ into an equivalent formula $\phi'$ in Gaifman
normal form such that each of the local formulas in $\phi'$ belongs to
$\LocFOplus[p,q]$ and each formula $\psibl(x)$ of a basic local sentence
in $\phi'$ belongs to $\LocFOplus[p{-}1,q]$
(Theorem~\ref{thm:rpGaifman}).

Section~\ref{sec:FixingThm71} presents a refined version of
Theorem~\ref{thm:rpGaifman}, which
goes beyond Theorem~\ref{thm:rpGaifman} in two ways, both of which are important
for applications such as counting and enumeration: first, the 
normal form is a disjunction of mutually exclusive formulas, and second, the free
variables in the local formulas are always interpreted by elements
that are close to each other in the Gaifman graph of the structure.
Specifically, Theorem~\ref{thm:forGroheS18} fixes the flaw in \cite[Theorem~7.1]{GroheS18}.

Section~\ref{sec:rank-new} introduces a new rank measure $\rk(\phi)$ for
formulas $\phi$ in $\FOplus$. As a corollary to
Theorem~\ref{thm:rpGaifman}, we obtain a rank-preserving Gaifman
normal form theorem
for $\FOplus$, which enables us to transform every $\phi$ in $\FOplus$
of rank $\rankletter\deff\rk(\phi)$ into an equivalent $\FOplus$ formula $\phi'$ in
Gaifman normal form such that each of the local formulas in $\phi'$
as well as each formula $\psibl(x)$ of a basic local sentence in $\phi'$   
is $r$-local and has rank at most $\rankletter$
(Theorem~\ref{thm:rankGaifman}).
Here, $r$ can be chosen to be a particular number that only depends on~$\rankletter$.

Section~\ref{sec:app}, as an application of our rank-preserving Gaifman normal form theorem,
provides a simplified proof of the main result of \cite{GroheKS17},
i.e., it shows that model-checking of first-order sentences on nowhere
dense classes of structures can be done in time $O(n^{1+\epsilon})$
for every $\epsilon>0$.

Section~\ref{sec:conclusion}
provides a short summary and points out an open research question.

\section{Preliminaries}
\label{sec:prel}
We denote the set of nonnegative integers by $\Nat$ and the set of
positive integers by $\PNat$. For $k,\ell\in \NN$ we let
$[k,\ell]\deff\setc{i\in\NN}{k\leq i\leq \ell}$, and we let
$[k]\deff[1,k]$.

We denote tuples by bold-face letters, as in $\vec x=(x_1,\ldots,x_k)$. For a subset $I\subseteq[k]$ of the index set, by $\vec x_I$ we denote the projection of $\vec x$ on the indices in $I$, i.e., if $I=\{i_1,\ldots,i_\ell\}$ with $i_1<\cdots<i_\ell$ then $\vec x_I\coloneqq(x_{i_1},\ldots,x_{i_\ell})$.

We assume that reader is familiar with basic notions of graph theory.
Graphs are always undirected and simple (no loops or
parallel edges) in this paper.
For a graph $G$ we write $V(G)$ and $E(G)$ to denote the set of
vertices and the set edges of $G$. For graphs $G$ and $H$ we write
$H\subseteq G$ and, equivalently, $G\supseteq H$ to indicate that $H$
is a subgraph of $G$. 
For a set $X\subseteq V(G)$, we denote the subgraph induced by $G$ on $X$ by $G[X]$, and we
let $G\setminus X\coloneqq G[V(G)\setminus X]$.

\subsection{Relational Structures}
A \emph{vocabulary} is a finite set $\sigma$ of relation symbols, each
relation symbol $R$ coming with an arity $\ar(R)\in\Nat$. Note that we
admit $0$-ary relation symbols, which is convenient in some
applications (cf.\ e.g.\ Section~\ref{sec:app}
and \cite{GroheS18}).
A \emph{$\sigma$-structure} is a tuple
$\CA=\big(A,\big(R^\CA\big)_{R\in\sigma}\big)$, where $A$, the \emph{universe} of $\CA$, is a non-empty set  and 
$R^\CA\subseteq A^{\ar(R)}$ for every $R\in\sigma$. 
We define
\emph{induced substructures} (denoted $\CA[X]$)
in the usual
way, i.e., for a non-empty set $X\subseteq A$, we let $\CA[X]$ be the
$\sigma$-structure with universe $X$ and relations $R^{\CA[X]}\deff
R\cap X^{\ar(R)}$ for every $R\in\sigma$.
For a vocabulary $\sigma'\supseteq \sigma$, a
\emph{$\sigma'$-expansion} of a $\sigma$-structure $\CA$ is a
$\sigma'$-structure $\CA'$ with $A'=A$ and $R^{\CA'}=R^{\CA}$ for all $R\in\sigma$.
The \emph{order} $|\CA|$ of a
structure $\CA$ is the cardinality of its universe $A$.

An \emph{isomorphism} between two $\sigma$-structures $\CA$ and $\CB$
is a bijective mapping $\pi\colon A\to B$ such that for all
$R\in\sigma$, for $r\deff \ar(R)$ and all tuples $(a_1,\ldots,a_r)\in
A^r$ we have $(a_1,\ldots,a_r)\in R^{\CA} \iff
(\pi(a_1),\ldots,\pi(a_r))\in R^{\CB}$. 

The \emph{Gaifman graph} $G_{\CA}$ of a $\sigma$-structure $\CA$ is
the graph with vertex set $A$ and edges $\set{a,b}$ for all
distinct $a,b\in A$ such that $a,b\in\{a_1,\ldots,a_{\ar(R)}\}$ for
some tuple $(a_1,\ldots,a_{\ar(R)})\in R^{\CA}$ for some
$R\in\sigma$. We generalise graph-theoretic notions such as
connectivity and degree to $\sigma$-structures via their Gaifman
graphs. Most importantly, the distance $\dist^\CA(a,b)$ between two
elements $a,b\in A$ is the length (i.e., the number of edges) of a shortest path from $a$ to $b$ in the Gaifman graph $G_\CA$, or $\infty$ if no such path exists.

For every $r\in\Nat$, the \emph{$r$-neighbourhood} of an element $a\in
A$ in a $\sigma$-structure $\CA$ is the set $N_r^\CA(a)\coloneqq\{b\in
A\mid\dist^\CA(a,b)\le r\}$. More generally, for a $k\in\NNpos$ and a tuple $\vec
a=(a_1,\ldots,a_k)\in A^k$, we let $N_r^\CA(\vec a)\coloneqq\bigcup_{i\in[k]}N^\CA_r(a_i)$. Moreover, the \emph{$r$-neighbourhood structure of $\vec a$} is the induced substructure $\CN_r^\CA(\vec a)\coloneqq\CA\big[N_r^\CA(\vec a)\big]$.

\emph{For the rest of this paper, we fix a
  vocabulary $\sigma$. Unless explicitly stated otherwise, all structures and all logical formulas are of vocabulary $\sigma$.}

\subsection{First-order logic $\FO$}
We assume that the reader is familiar with first-order logic $\FO$ (for background, cf.\ \cite{EbbinghausFT21}). The set of free variables of a formula $\phi$ is denoted by $\free(\phi)$. We write $\phi(x_1,\ldots,x_k)$ to stipulate
that $\free(\phi)\subseteq\{x_1,\ldots,x_k\}$. For a structure $\CA$
and elements $a_1,\ldots,a_k$ we write
$\CA\models\phi[a_1,\ldots,a_k]$ if $\CA$ with the variable assignment
$x_i\mapsto a_i$ for $i\in [k]$ satisfies $\phi$. The notation
$\phi(x_1,\ldots,x_k)$ is also convenient for substitutions; by
$\phi(y_1,\ldots,y_k)$ we denote the formula obtained from $\phi$ by
substituting $x_i$ with $y_i$, for all $i$ (and renaming bound
variables if necessary, cf.~\cite[Section~III.8]{EbbinghausFT21}).
The \emph{quantifier rank} $\qr(\phi)$ of a formula $\phi$ is defined in the usual way as the maximum nesting depth of quantifiers in $\phi$.

\subsection{The Ehrenfeucht-\Fraisse\  Game for $\FO$}
Let $\CA,\CB$ be $\sigma$-structures. A \emph{local
  isomorphism} from $\CA$ to $\CB$ is a mapping $\pi$
with domain $\dom(\pi)\subseteq A$ and range $\rg(\pi)\subseteq B$
that is an isomorphism between the induced substructures $\CA[\dom(\pi)]$ and $\CB[\rg(\pi)]$. 

Let $k,p\in\Nat$ and $\vec a=(a_1,\ldots,a_k)\in A^k,\vec
b=(b_1,\ldots,b_k)\in B^k$. The \emph{$p$-round Ehrenfeucht-\Fraisse\ game} on $(\CA,\vec a)$ and
$(\CB,\vec b)$ (for short: \emph{$p$-round EF game} on $(\CA,\vec a)$,
$(\CB,\vec b)$) is played by two players, \emph{Spoiler} and
\emph{Duplicator}, in $p$ rounds. For each $\ell\in[0,p]$, the
\emph{position} of a play after round $\ell$ is a pair $\vec
a^\ell,\vec b^\ell$, where $\vec a^\ell=(a_1,\ldots,a_{k+\ell})$
consists of the elements $a_1,\ldots,a_k$ of the tuple $\vec a$ and
additional elements
$a_{k+1},\ldots,a_{k+\ell}\in A$, and similarly $\vec b^\ell=(b_1,\ldots,b_{k+\ell})$
consists of the elements $b_1,\ldots,b_k$ of the tuple $\vec b$ and
additional elements $b_{k+1},\ldots,b_{k+\ell}\in B$.
In round $\ell{+}1$, for $\ell\in [0,p{-}1]$, Spoiler picks
\begin{itemize}
\item either an $a_{k+\ell+1}\in A$, and Duplicator answers by picking
  a $b_{k+\ell+1}\in B$,
\item or a $b_{k+\ell+1}\in B$, and Duplicator answers by picking
  an $a_{k+\ell+1}\in A$.
\end{itemize}
Duplicator wins the game if the mapping
$a_j\mapsto b_j$ for $j\in[k{+}p]$ is a local
isomorphism.

The well-known \emph{Ehrenfeucht-\Fraisse\ Theorem} (cf., e.g.,
\cite{EbbinghausF99,Libkin04}) states that Duplicator has a winning
strategy for the $p$-round EF game on $(\CA,\vec a)$,
    $(\CB,\vec b)$ if and only if 
  for all $\FO$ formulas $\phi(x_1,\ldots,x_k)$ of quantifier rank
  $\qr(\phi)\leq p$ we have
    $
      \CA\models\phi[\vec a]\iff\CB\models\phi[\vec b].
    $

\subsection{First-order logic with distance atoms $\FOplus$}\label{subsec:FOplus}
Note that for every $d\in\Nat$ there is an $\FO$ formula
$\dist_{\leq d}(x,y)$ such that for all
$\sigma$-structures $\CA$ and $a,b\in
A$ we have $\CA\models\dist_{\leq d}[a,b]\iff\dist^\CA(a,b)\le
d$. The formula $\dist_{\leq d}(x,y)$ depends on the vocabulary
$\sigma$ and can be chosen to be 
of quantifier rank $\lceil \log_2 d \rceil + \alpha-2$
for $\alpha\deff\max\setc{\ar(R)}{R\in\sigma}$, if $d\geq 1$ and
$\alpha\geq 2$, and of quantifier rank 0 if $d=0$ or $\alpha\leq 1$.
We will use the extension $\FOplus$ of $\FO$
that has additional atomic formulas
\[
  \dist(x,y)\le d\,,
\]
  for all variables $x,y$ and all $d\in\Nat$, with the obvious
  semantics. We call these the \emph{distance atoms}. The
  \emph{maximum distance}
  of an $\FOplus$ formula $\phi$ is the least $d\in\Nat$ such that for all distance atoms $\dist(x,y)\le d'$ appearing in $\phi$ it holds that $d'\le d$. 

Note that $\FOplus$ is a purely syntactic extension of $\FO$: if we
replace each distance atom $\dist(x,y)\le d$ in an $\FOplus$ formula
by the formula $\dist_{\leq d}(x,y)$, we obtain an equivalent $\FO$ formula. 
The point of this syntactic extension is that we treat the distance atoms really as atomic formulas rather than just abbreviations. In particular, they do not contribute to the quantifier rank, and they are not affected by the following normalisation procedure.

\subsection{Normalisation}\label{sec:normalisation}
Let us fix a total order of the relation symbols in $\sigma$. We also
fix a total order on the (countable) set of variables. Assume that
$\vvar_1,\vvar_2,\vvar_3,\ldots{}$ is an enumeration of the variables in this
order. We view $\FOplus$ formulas as strings over an alphabet
consisting of the symbols in $\sigma$, the variables $\vvar_i$ for
$i\in\NNpos$, the symbols $\dist$ and $\le$, a symbol $d$
for each $d\in \Nat$, symbols for the Boolean connectives and quantifiers, and parentheses and commas. Ordering this alphabet in some way based on the orders of $\sigma$ and the variables, we can order formulas lexicographically.
By induction on the quantifier rank, we define a \emph{normalised formula} $\tilde\phi$ for each $\FOplus$ formula $\phi$.  If $\phi$ is quantifier-free, then $\tilde\phi$ is obtained by transforming $\phi$ into conjunctive normal form, ordering the literals in each clause lexicographically and eliminating duplicates, then ordering the clauses lexicographically and eliminating duplicates. To make this procedure deterministic, we must define the conjunctive normal form in a canonical way. It is not important how exactly we do this. 

If $\phi$ has quantifier rank $p\ge 1$, then it is a Boolean
combination of atomic formulas and formulas $Qx\,\psi$ where
$\qr(\psi)\le p{-}1$ and $Q\in\{\exists,\forall\}$ and $x$ is a
variable. Towards defining $\tilde\phi$, the first thing we do is
normalise each of the formulas $\vartheta\coloneqq Qx\,\psi$. We choose the
minimum $i$ such that $\vvar_i\not\in\free(\vartheta)$, and we
let $\psi'$ be the formula obtained from $\psi$ by substituting
$x_i$ for $x$. Then, $\vartheta$ is equivalent to
$Q\vvar_i\,\psi'$. 
Afterwards, we normalise $\psi'$ by induction, and we let $\tilde\vartheta\coloneqq Q\vvar_i\,\tilde\psi'$. 
In the general case that $\phi$ is a Boolean combination of atomic
formulas and quantified formulas $\vartheta$ of the form $Qx\,\psi$, we first replace each quantified formula $\vartheta$ by $\tilde\vartheta$ and then proceed as for quantifier-free formulas, i.e., bring the Boolean combination into conjunctive normal form, order and eliminate duplicates.

\begin{remark}\label{rem:norm-local}
There is one technical issue we need to attend to: we want our normal
form to be compatible with restricted quantification of the form
$\exists y(\dist(x,y)\le d\wedge\psi)$. Therefore, for a formula
$\phi$ of the form $\exists y\big(\dist(x,y)\le d\wedge\psi(x,y,\vec
z)\big)$, we define $\tilde\phi$ to be the formula $\exists
\vvar_i\big(\dist(x,\vvar_i)\leq d\wedge\tilde\psi(x,\vvar_i,\vec
z)\big)$ obtained by first substituting the bound variable $y$ by the variable $\vvar_i$
for the minimum $i$ such that $\vvar_i\not\in\free(\phi)$
and then normalising $\psi(x,\vvar_i,\vec z)$.
\uend
\end{remark}

\noindent
Observe that $\tilde\phi$ has the same set of free variables, the same quantifier rank, and the same maximum distance as $\phi$. Furthermore, if $\phi$ is an $\FO$ formula (i.e., contains no distance atoms) then $\tilde\phi$ is an $\FO$ formula as well.

A \emph{normalised formula} is a formula $\tilde\phi$ for some 
$\FOplus$ formula $\phi$. Observe that for every normalised formula $\tilde\phi$ it holds that 
$\tilde{\tilde\phi}=\tilde\phi$.
It is easy to see that for every (finite) vocabulary $\sigma$, every
fixed finite set $X$ of variables, and all $p,d\in\Nat$ there are only
finitely many normalised formulas of vocabulary $\sigma$, quantifier
rank at most $p$, maximum distance at most $d$, and with all free variables in $X$; and the set of all these formulas is computable from $\sigma,p,d$, and $X$.

\section{Gaifman's Theorem}
\label{sec:Gaifman}

Let $r\geq 0$.
For a formula $\phi$ and a non-empty tuple $\vec z=(z_1,\ldots,z_k)$
of variables with $\free(\phi)\subseteq\set{z_1,\ldots,z_k}$, 
we say that $\phi(\vec z)$ is \emph{$r$-local} (around $\vec z$) if for all
$\sigma$-structures $\CA$ and all $\vec a=(a_1,\ldots,a_k)\in A^k$ we have
\begin{equation}\label{eq:DefLocal}
  \CA\models\phi[\vec a] \ \ \iff \ \ \CN_r^\CA(\vec a)\models\phi[\vec a];
\end{equation}
and $\phi(\vec z)$ is \emph{local} if it is $r$-local for some $r\geq 0$.

A \emph{basic local sentence} is an $\FO$ sentence of the form
\begin{equation}\label{eq:basic-local-Gaifman}
  \exists x_1\cdots\exists x_m\;\Bigl(\bigwedge_{1\le i<j\le m}\dist(x_i,x_j)>2r\;\wedge\;\bigwedge_{1\le i\le m}\psibl(x_i)\Bigr),
\end{equation}
where $r\geq 0$,
$m\in\PNat$, and $\psibl(x)$ is an $r$-local $\FO$ formula with one free variable $x$.
Here $\dist(x_i,x_j)>2r$ abbreviates the $\FO$ formula
$\neg\dist_{\leq \lfloor 2r \rfloor}(x_i,x_j)$; this condition ensures that the
$r$-neighbourhoods $N_r^\CA(x_i)$ and $N_r^\CA(x_j)$ are pairwise disjoint for $i\ne j$.

An $\FO$ formula is said to be in \emph{Gaifman normal form} if it is a
Boolean combination of local formulas and basic local sentences.
Gaifman's Theorem reads as follows.

\begin{theorem}[Gaifman \cite{Gaifman82}]\label{thm:gaifman}
  Every $\FO$ formula $\phi$ is equivalent to a
  formula $\phi'$ in Gaifman normal form.
  Moreover, $\free(\phi')=\free(\phi)$, and there is an algorithm
  that, upon input of $\phi$, 
  computes such a $\phi'$.
\end{theorem}

We fix some more notation.
Consider a basic local sentence $\xi$ of the form \eqref{eq:basic-local-Gaifman}. 
The \emph{inner quantifier rank} of $\xi$ is the quantifier rank of
$\psibl(x)$,  the \emph{width} of $\xi$ is $m$, and the \emph{radius} of
$\xi$ is $r$.

Let $\phi'$ be a formula in Gaifman normal form, and let $S$ and $L$
be the sets of basic local sentences and of local formulas,
respectively, such that $\phi'$ is a Boolean combination of the
formulas in $S\cup L$.
The \emph{inner quantifier rank} (the \emph{width}, respectively) of  $\phi'$
is the maximum of the
inner quantifier ranks (the widths, respectively) of all sentences in $S$.
The \emph{outer quantifier rank} of $\phi'$ is the maximum of the
quantifier ranks of all the formulas in $L$.
The \emph{radius} of $\phi'$ is the minimum $r\in\Nat$ such that all
sentences in $S$ have radius  $\leq r$ and all formulas in $L$
 are $r$-local. To make the distinction between the different quantifier ranks clearer, we refer to the quantifier rank of $\phi'$ as the \emph{total quantifier rank}.
Observe that if $\phi'$ has outer quantifier rank $q'$, inner
quantifier rank $q$,  width $m$, and radius $r$ then $\phi'$ has total
quantifier rank at most $\max\{q',m{+}q,m{+}\ell\}$,
where $\ell=\lceil
\log_2(2r)\rceil + O(1)$ is the quantifier rank needed for expressing
``$\dist(x,y)>2r$'' by a first-order formula (the $O(1)$-term depends on
the vocabulary $\sigma$, cf.\  the  beginning of Section~\ref{subsec:FOplus}).

The following example shows that Gaifman's normal form theorem does not preserve
the outer quantifier rank: it provides an $\FO$ formula $\phi(x)$ of quantifier rank $2$
that is not equivalent to any Boolean combination
of local formulas of quantifier rank at most~$2$ and
of basic local sentences.

\begin{example}\label{ex:1}
  Let $\sigma=\set{E,\mathsf{red},\mathsf{blue}}$
  with $\ar(E)=2$ and $\ar(\mathsf{red})=\ar(\mathsf{blue})=1$,
  and consider the $\FO$ formula
  \[
  \phi(x) \ \ \coloneqq \ \ \exists y \Big(\mathsf{red}(y)\wedge\neg\exists z\big(E(x,z)\wedge E(z,y)\big)\wedge\neg\exists z\big(\mathsf{blue}(z)\wedge E(y,z)\big)\Big).
\]
  This formula states that there exists a red node $y$ which is not
  reachable from $x$ by a walk
  of length 2
  and which does not have a blue neighbour.
  
\begin{figure}
  \centering
  \begin{tikzpicture}[
    scale=0.9,
    bvert/.style={circle, fill=black, draw=black, inner sep=1pt, minimum size=16pt,font=\footnotesize,text=white},
    rvert/.style={circle, fill=red!80!black, draw=black, inner sep=1pt, minimum size=16pt,font=\footnotesize,text=white},
    blvert/.style={circle, fill=blue!80!black, draw=black, inner sep=1pt, minimum size=16pt,font=\footnotesize,text=white},
  ]

  \fill[orange!50] (-7:7cm) arc[radius=7cm, start angle = -7, end angle = 187];
  \node[orange!80!black] at (5.5,6.1)
    {$N^\CG_7(v) = N^\CG_7(v')$};

  \draw[thick] (-2,0) -- (5,0);
  \draw[thick] (-2,0) -- (5,0);
  \draw[thick] (-2,0) -- (0,1.5) -- (2,0);
  \draw[thick] (0,1.5) -- (0,7.5);
  \draw[thick] (0,1.5) -- (0,7.5);

  \foreach \x in {1,2,4,5,6,8} {
    \node[bvert] at (\x-3, 0) {$\x$};
  }
  \node[below=17pt,anchor=base] at (-2,0) {$v$};
  \node[below=17pt,anchor=base] at (2,0) {$v'$};

  \node[rvert] at (0, 0) {$3$};
  \node[rvert] at (4, 0) {$7$};

  \node[bvert] at (-1, 0.75) {$9$};
  \node[bvert] at (1, 0.75) {$10$};
 
  \foreach \y in {11,...,15} {
     \node[bvert] at (0, \y-9.5) {$\y$};
  }

  \node[rvert] at (0,6.5) {$16$};
  \node[blvert] at (0,7.5) {$17$};
  \node[rvert] at (0,6.5) {$16$};
  \node[blvert] at (0,7.5) {$17$};

  \end{tikzpicture}
  \caption{Structure $\CG=\CG_7$ of Example~\ref{ex:1}. 
   Undirected edges represent directed edges in both directions, and
   $\mathsf{red}^{\CG}=\set{3,7,16}$ and $\mathsf{blue}^{\CG}=\set{17}$.
  }
  \label{fig:ex1}
\end{figure}

  Let $\CG$ be the $\sigma$-structure  depicted in Figure~\ref{fig:ex1}, and let $v\coloneqq 1$ and $v'\coloneqq 5$ (as indicated in Figure~\ref{fig:ex1}).
  Clearly, $\CG\models\phi[v]$ and $\CG\not\models\phi[v']$.
  Observe that $N_7^{\CG}(v)=N_7^\CG(v')$,
  and let $\CN_7\coloneqq\CN_7^\CG(v)=\CN_7^\CG(v')$. In Figure~\ref{fig:ex1}, $\CN_7$
  is the substructure of $\CG$ induced on the orange part.

  Now, consider an arbitrary $r\in\NN$ with $r\geq 7$, and
  let $\CG_r$ be the $\sigma$-structure that looks exactly like $\CG$,
  with the only difference that the
  path from vertex $11$ to the top red vertex (vertex $16$ for $r=7$) has length $r{-}2$  (hence, the 
  distance between the top red vertex and $v$ as well as $v'$ is
  $r$). Note that the structure $\CG$ depicted in Figure~\ref{fig:ex1}
  is $\CG_7$.
  It is easy to see that $\CG_r\models\phi[v]$ and $\CG_r\not\models\phi[v']$.
  Furthermore, $N_r^{\CG_r}(v)=N_r^{\CG_r}(v')$. Let
  $\CN_r\deff\CN_r^{\CG_r}(v) = \CN_r^{\CG_r}(v')$.
  Below, we will prove the following claim.
  
  \begin{claim}\label{claim:GaifmanCounterExample}
    For every $r\in\NN$ with $r\geq 7$, 
    Duplicator has a winning strategy for the 2-round
    Ehrenfeucht-\Fraisse\ game on $(\CN_r,v)$ and $(\CN_r,v')$.    
  \end{claim}  

  Combining this claim with
  the Ehrenfeucht-\Fraisse\ Theorem yields that for all $\FO$ formulas
  $\psi(x)$ of quantifier rank at most $2$ we have $\CN_r\models\psi[v]\iff\CN_r\models\psi[v']$. 
  Thus, for all $r$-local $\FO$ formulas $\psi(x)$ of
  quantifier rank at most $2$ we have
  $\CG_r\models\psi[v]\iff\CG_r\models\psi[v']$.

  This implies that $\phi(x)$ is not equivalent to a Boolean
  combination $\phi'(x)$ of local formulas
  of quantifier rank at most $2$ and of basic local sentences.
  Indeed, suppose for contradiction it was.
  Then let $L$ and $S$ be the sets of  local
  formulas and basic local sentences, respectively, such that $\phi'(x)$ is a Boolean
  combination of the formulas in $L\cup S$.
  Let $r\in\NN$ with $r\geq 7$ be such that each formula in $L$ is $r$-local. Then for
  each $\psi(x)$ in $L$ we have
  $\CG_r\models\psi[v]\iff\CG_r\models\psi[v']$. This implies that
  $\CG_r\models\phi'[v]\iff\CG_r\models\phi'[v']$. But this contradicts
  $\phi'(x)$ being equivalent to $\phi(x)$, because we know that
  $\CG_r\models\phi[v]$ and $\CG_r\not\models\phi[v']$.
  All that still remains to be done is to prove
  Claim~\ref{claim:GaifmanCounterExample}.

  \begin{proof}[Proof of Claim~\ref{claim:GaifmanCounterExample}]
    Let us fix an $r\ge 7$ and let $\CN\coloneqq\CN_r$. We assume that the vertex set of $\CN$ is $N\coloneqq\{1,\ldots,r{+}9\}$, where the vertices are numbered as indicated in Figure~\ref{fig:ex1} for $r=7$; for general $r$, the vertices on the path from node $11$ to the top red node are $11,\ldots,r{+}9$.

    We call a pair $(w,w')\in N^2$ \emph{good} if Duplicator has a
    winning strategy for the 1-round EF game on $(\CN,(v,w))$, $(\CN,(v',w'))$. 
    We need to prove that for every $w\in N$ there is a $w'\in N$
    such that $(w,w')$ is good, and for every $w'\in N$ there is
    a $w\in N$
    such that $(w,w')$ is good.

    We give a list of good pairs witnessing this; for each pair on the list it is straightforward to verify that it is good.
    \begin{itemize}
    \item $(x,x)$ for each $x\ge 11$;
      \item $(1,5),(9,10),(2,4),(3,3)$ and the reversed pairs $(5,1),(10,9),(4,2)$;
      \item $(6,15), (7,16)$;
      \item $(2,6),(3,7), (8,8)$.
    \end{itemize}
    Note that indeed every $x\in N$ occurs in the first and in the
    second component of a some pair in this list. This completes the
    proof of Claim~\ref{claim:GaifmanCounterExample}.
  \end{proof}
\end{example}

\section{The Subset $\FOplus[p,q]$ of $\FOplus$}\label{sec:FOpluspq}
Recall that we fixed a vocabulary $\sigma$ and only consider formulas
and structures of this vocabulary.
Throughout the rest of this paper we will use an arbitrary but fixed
(computable)
function
$\myrho\colon\Nat^2\to\Nat$ that satisfies the following constraints for
all $p,q,q'\in\Nat$ with $p\leq q\leq q'$:
\begin{enumerate}
\item\label{myrho:constraint:zero}
  $\myrho(0,q)\geq 1$, 
\item\label{myrho:constraint:inc}
  $\myrho(p,q)\geq \myrho(p{-}1,q) \cdot \big(2+4(q{-}p)\big) $, for
  $p\geq 1$,
  \item\label{myrho:constraint:q}
    $\myrho(p,q)\leq\myrho(p,q')$,  \ and, for $p\geq 1$, \ 
    $\myrho(p,q)-\myrho(p{-}1,q)\leq \myrho(p,q')-\myrho(p{-}1,q')$.
\end{enumerate}  
\begin{example}\label{example:myrho}
The conditions \ref{myrho:constraint:zero}--\ref{myrho:constraint:q} are satisfied when choosing $\myrho$ to be
any of the following functions.
\begin{itemize}
\item $\myrho_1(p,q)\deff 2^p \cdot
  \prod_{j\in[q-p,q-1]}(2j{+}1)$.

  Then, $\myrho_1(0,q)=1$, and for $p\geq 1$ we have
  $\myrho_1(p,q)=\myrho_1(p{-}1,q)\cdot \big(2+ 4(q{-}p)\big)$.
  In particular, conditions \ref{myrho:constraint:zero} and
  \ref{myrho:constraint:inc} are satisfied. Concerning condition
  \ref{myrho:constraint:q}, $\myrho_1(p,q)\leq \myrho_1(p,q')$
  obviously holds for all $p,q,q'\in\NN$ with $p\leq q\leq q'$.
  Furthermore, for each $q''\in\set{q,q'}$ we have
  $\myrho_1(p,q'')-\myrho_1(p{-}1,q'')=\myrho_1(p{-}1,q'')\cdot
  \big(1+4(q''{-}p)\big)$.
  Thus, using $q\leq q'$ and $\myrho_1(p{-}1,q)\leq
  \myrho_1(p{-}1,q')$, we obtain that
  $\myrho_1(p,q)-\myrho_1(p{-}1,q)\leq
  \myrho_1(p,q')-\myrho_1(p{-}1,q')$.
  Hence, $\myrho_1$ also satisfies condition \ref{myrho:constraint:q}.
  
\item $\myrho_2(p,q)\deff (4q)^p$.

  Then, $\myrho_2(0,q)=1$, and for $p\geq 1$ we have $\myrho_2(p,q)=\myrho_2(p{-}1,q)\cdot 4q >
  \myrho_2(p{-}1,q)\cdot \big(2+4(q{-}p)\big)$.
  In particular, conditions \ref{myrho:constraint:zero} and
  \ref{myrho:constraint:inc} are satisfied. Concerning condition
  \ref{myrho:constraint:q}, $\myrho_2(p,q)\leq \myrho_2(p,q')$
  obviously holds for all $p,q,q'\in\NN$ with $p\leq q\leq q'$.
  Furthermore, for each $q''\in\set{q,q'}$ we have
  $\myrho_2(p,q'')-\myrho_2(p{-}1,q'')=\myrho_2(p{-}1,q'')\cdot (4q{-}1)$.
  Thus, using $q\leq q'$ and $\myrho_2(p{-}1,q)\leq
  \myrho_2(p{-}1,q')$, we obtain that
  $\myrho_2(p,q)-\myrho_2(p{-}1,q)\leq
  \myrho_2(p,q')-\myrho_2(p{-}1,q')$.
  Hence, $\myrho_2$ also satisfies condition \ref{myrho:constraint:q}.

\item $\myrho_3(p,q)\deff (4q)^{q+p}$.

  Then, $\myrho_3(0,q)=(4q)^q\geq 1$, and for $p\geq 1$ we have
  $\myrho_3(p,q)=\myrho_3(p{-}1,q)\cdot 4q > \myrho_3(p{-}1,q)\cdot
  \big(2+4(q{-}p)\big)$.
  Thus, conditions \ref{myrho:constraint:zero} and
  \ref{myrho:constraint:inc} are satisfied. Condition
  \ref{myrho:constraint:q} can be verified for $\myrho_3$ in the same
  way as for $\myrho_2$.
  \uend
\end{itemize}  
\end{example}  

\begin{definition}\label{def:FOpluspq}
For $p,q\in\Nat$ with $p\leq q$, we define the following sets
$\FOplus[p,q]$:
\begin{itemize}
\item $\FOplus[0,q]$ is the set of all formulas $\phi$ such that
  $|\free(\phi)|\le q$
  and $\phi$ is a Boolean combination of atomic $\FO$ formulas and
  distance atoms $\dist(x,y)\le d$ with $d\le \myrho(0,q)$.
\item For $p>0$, $\FOplus[p,q]$ is the set of all formulas $\phi$ with
  $|\free(\phi)|\le q-p$
   such that
   $\phi$ is a Boolean combination of formulas in $\FOplus[p{-}1,q]$, distance atoms $\dist(x,y)\le d$ with
  $d\le \myrho(p,q)$, formulas $\exists y\, \psi$
  where $\psi\in\FOplus[p{-}1,q]$,
  and formulas
  \begin{equation}
    \label{eq:1}
    \exists y\; \big(\dist(x,y)\le
    d\ \wedge \psi\big),
  \end{equation}
  where
  $\psi\in\FOplus[p{-}1,q]$
  and $d\le \myrho(p,q)-\myrho(p{-}1,q)$.
  \uend
\end{itemize}
\end{definition}

Observe that $\FOplus=\bigcup_{p,q\in\Nat,0\le p\le q}\FOplus[p,q]$.
Furthermore, every $\phi\in\FOplus[p,q]$ has quantifier rank
$\qr(\phi)\leq p$ and has at most $q{-}p$ free variables.

\begin{remark}\label{remark:DefOfFOpq}
\begin{mea}
\item\label{item:remark:DefOfFOpq:q}
  For $p,q,q'\in\NN$ with $p\leq q\leq q'$ we have $\FOplus[p,q]\subseteq\FOplus[p,q']$.
  This easily follows from the fact that $\myrho$ satisfies condition \ref{myrho:constraint:q}.
\item\label{item:remark:DefOfFOpq:p}
  Let $p,p',q\in\NN$ with $p<p'\leq q$. Then,
$\FOplus[p,q]\nsubseteq\FOplus[p',q]$, because
$\FOplus[p,q]$ contains formulas with exactly $q{-}p$ free variables
--- and these formulas do not belong to $\FOplus[p',q]$.
But every formula $\phi$ in $\FOplus[p,q]$ that
satisfies
$|\free(\phi)|\leq q-p'$ also belongs to $\FOplus[p',q]$ (this
immediately follows from Definition~\ref{def:FOpluspq}).
\uend
\end{mea}
\end{remark}

\begin{remark}
  The definition of $\FOplus[p,q]$ is related to the
notion ``$q$-rank at most $\ell$'' defined in
\cite{GroheKS17,GroheS18} in the following way. If we did not admit
formulas of the form \eqref{eq:1} in the definition of $\FOplus[p,q]$
(or only such formulas with $d\le\myrho(p{-}1,q)$) and used the
particular function
$\myrho(p,q)\coloneqq (4q)^{q+p}$ (cf., Example~\ref{example:myrho}),
which is the function $f_q(p)$ of
\cite{GroheKS17,GroheS18}, then $\FOplus[p,q]$ would be exactly the
class of formulas $\phi$ of $q$-rank at most $p$, as defined in
\cite{GroheKS17,GroheS18},
where $|\free(\phi)|\leq q{-}p$. 
\uend
\end{remark}

We often need to work with normalised formulas (cf.\ 
Section~\ref{sec:normalisation}). When normalising, we need to pay
attention to the free variables; we will often assume them to be in an
initial segment $x_1,\ldots,x_k$ of the canonical enumeration of the
variables. We define $\normFOplus[p,q]$ to be the set of normalised
formulas
$\tilde\phi$ for all $\phi\in\FOplus[p,q]$ with
$\free(\phi)\subseteq\{x_1,\ldots,x_{q-p}\}$. Note that for all $p,q$
the set $\normFOplus[p,q]$ is finite.  
Using Remark~\ref{rem:norm-local}, one obtains that for every
$\phi\in\FOplus[p,q]$ with
$\free(\phi)\subseteq\set{x_1,\ldots,x_{p-q}}$, the normalised formula
$\tilde\phi$ also belongs to
$\FOplus[p,q]$. Hence,
$\normFOplus[p,q]\subseteq\FOplus[p,q]$.

\subsection{Types}

Let $k,p,q\in\Nat$ with $k+p\leq q$.
A \emph{$(k,p,q)$-type}
is a set $\Theta\subseteq\normFOplus[p,q]$
such that for each
$\phi(x_1,\ldots,x_k)\in\FOplus[p,q]$, either $\normalised{\phi}\in\Theta$ or
$\normalised{\neg\phi}\in\Theta$. Recall that the notation $\phi(x_1,\ldots,x_k)$ stipulates $\free(\phi)\subseteq\{x_1,\ldots,x_k\}$ and thus $\tilde\phi\in\normFOplus[p,q]$.
Since $\normFOplus[p,q]$ is finite, it is straightforward
to see that there is a formula $\psi(x_1,\ldots,x_k)\in\normFOplus[p,q]$ such that for all $\sigma$-structures
$\CA$ and tuples $\vec a\in A^k$ we have
\[
  \CA\models\psi[\vec a]\ \ \iff \ \ \CA\models\theta[\vec a]\text{ for all
  }\theta\in\Theta.
\]
Indeed, we can take $\psi$ to be the normalisation of the conjunction of all
formulas in 
$\Theta$,
and we will henceforth denote $\psi$ as $\bigwedge\Theta$. Note that
$\free(\psi)\subseteq\{x_1,\ldots,x_k\}$. It will sometimes be
convenient to use different variable names, say, $y_1,\ldots,y_k$. We
write $\psi'(y_1,\ldots,y_k)\coloneqq\bigwedge\Theta$ to denote that
$\psi'$ is the formula obtained from $\psi\deff\bigwedge\Theta$ by
substituting $x_i$ with $y_i$ for all $i\in[k]$
(this includes, if necessary, a renaming of bound variables to ensure that
$\psi'$ has the same meaning for $(y_1\ldots,y_k)$ that
$\psi$ has for $(x_1,\ldots,x_k)$).
We use similar conventions for the local types that we will introduce later.

For a $\sigma$-structure $\CA$ and a tuple $\vec a=(a_1,\ldots,a_k)\in
A^k$, the \emph{$(p,q)$-type} of $\vec a$ in $\CA$ is the set
$\tp_{p,q}(\CA,\vec a)$ of
all formulas $\phi(x_1,\ldots,x_k)\in\normFOplus[p,q]$ such that
$\CA\models\phi[\vec a]$. Then for all $\sigma$-structures
$\CB$ and tuples $\vec b\in B^k$ we have
\begin{equation}\label{eq:tpVsFormulas}
  \begin{array}{cl}
&\tp_{p,q}(\CA,\vec a)=\tp_{p,q}(\CB,\vec b)
\medskip\\
   \iff
& \text{for all }\phi(x_1,\ldots,x_k)\in\FOplus[p,q]:\;\big(\CA\models\phi[\vec a]\iff\CB\models\phi[\vec
b]\big).
\end{array}
 \end{equation}
Also note that for $\theta\coloneqq\bigwedge\tp_{p,q}(\CA,\vec a)$ we have
\[
  \CB\models\theta[\vec b] \ \ \iff \ \ \tp_{p,q}(\CA,\vec
  a)=\tp_{p,q}(\CB,\vec b).
\]

\subsection{An Ehrenfeucht-\Fraisse\ Game for $\FOplus[p,q]$}\label{subsec:pqgame}
Let $\CA,\CB$ be $\sigma$-structures, and let $d\in\Nat$. A \emph{local
  $d$-isomorphism} from $\CA$ to $\CB$ is a mapping $\pi$
with domain $\dom(\pi)\subseteq A$ and range $\rg(\pi)\subseteq B$
such that (i) for all 
$a,a'\in\dom(\pi)$
either
$\dist^\CA(a,a')=\dist^\CB\big(\pi(a),\pi(a')\big)\leq d$ or
$\dist^\CA(a,a')>d$ and $\dist^\CB\big(\pi(a),\pi(a')\big)>d$, and
(ii) for all $R\in\sigma$, for $r\deff\ar(R)$ and all
$(a_1,\ldots,a_r)\in \dom(\pi)^{r}$ we have
$(a_1,\ldots,a_r)\in R^\CA\iff (\pi(a_1),\ldots,\pi(a_r))\in R^\CB$. Observe that condition
(i) implies that $\pi$ is bijective and then condition (ii) implies that $\pi$
is a local isomorphism in the usual sense.
Futhermore, if $\pi$ is a local $d$-isomorphism, then it also is a local
$d'$-isomorphism for every $d'\leq d$.

Let $k,p,q\in\Nat$ such that $k+p\le q$. Furthermore, let $\CA,\CB$ be
$\sigma$-structures and $\vec a=(a_1,\ldots,a_k)\in A^k,\vec
b=(b_1,\ldots,b_k)\in B^k$. The \emph{$(p,q)$-game} on $(\CA,\vec a)$,
$(\CB,\vec b)$ is played by two players, \emph{Spoiler} and
\emph{Duplicator}, in $p$ rounds. For each $\ell\in [0,p]$, the
\emph{position} of a play after round $\ell$ is a pair $\vec
a^\ell,\vec b^\ell$, where $\vec a^\ell=(a_1,\ldots,a_{k+\ell})$
consists of the elements $a_1,\ldots,a_k$ of the tuple $\vec a$ and
additional elements
$a_{k+1},\ldots,a_{k+\ell}\in A$, and similarly $\vec b^\ell=(b_1,\ldots,b_{k+\ell})$
consists of the elements $b_1,\ldots,b_k$ of the tuple $\vec b$ and
additional elements $b_{k+1},\ldots,b_{k+\ell}\in B$.
In round $\ell{+}1$, for $\ell\in[0,p{-}1]$, Spoiler decides if he
makes an \emph{$i$-near move} for some $i\in[k{+}\ell]$ or a \emph{far move}.
\medskip

\noindent
\textit{$i$-near move:} Spoiler picks
\begin{itemize}
\item either an $a_{k+\ell+1}\in A$ such that
  $d\coloneqq\dist^\CA(a_i,a_{k+\ell+1})\le \myrho(p{-}\ell,q)-\myrho(p{-}(\ell{+}1),q)$, and Duplicator answers by picking
  a $b_{k+\ell+1}\in B$ such that $\dist^\CB(b_i,b_{k+\ell+1})\le d$,
\item or a $b_{k+\ell+1}\in B$ such that
  $d\coloneqq\dist^\CB(b_i,b_{k+\ell+1}) \le \myrho(p{-}\ell,q)-\myrho(p{-}(\ell{+}1),q)$, and Duplicator answers by picking
  an $a_{k+\ell+1}\in A$ such that $\dist^\CA(a_i,a_{k+\ell+1})\le d$.
\end{itemize}
\smallskip

\noindent
\textit{Far move:} Spoiler picks
\begin{itemize}
\item either an $a_{k+\ell+1}\in A$ such that
  $\dist^\CA(a_i,a_{k+\ell+1})> \myrho(p{-}\ell,q)-\myrho(p{-}(\ell{+}1),q)$ for
  all $i\in[k{+}\ell]$, and Duplicator answers by picking
  a $b_{k+\ell+1}\in B$,
\item or a $b_{k+\ell+1}\in B$ such that
  $\dist^\CB(b_i,b_{k+\ell+1})> \myrho(p{-}\ell,q)-\myrho(p{-}(\ell{+}1),q)$ for
  all $i\in[k{+}\ell]$, and Duplicator answers by picking
  an $a_{k+\ell+1}\in A$.
\end{itemize}
Duplicator wins the game if for each $\ell\in[0,p]$ the mapping
$a_j\mapsto b_j$ for $j\in[k{+}\ell]$ is a local
$\myrho(p{-}\ell,q)$-isomorphism.

Note that in her answer to a far move, Duplicator does not have to
satisfy any distance constraints except for those implicitly imposed
by the winning condition.

In the following,
we argue that if Duplicator has a winning strategy for the
$(p{+}1,q)$-game on $(\CA,\vec a)$, $(\CB,\vec b)$, then she also has
a winning strategy for the $(p,q)$-game on $(\CA,\vec a)$, $(\CB,\vec b)$.
To verify this,
recall the conditions \ref{myrho:constraint:zero} and
\ref{myrho:constraint:inc} formulated at the beginning of
Section~\ref{sec:FOpluspq}. They imply
that $\myrho(p{+}1{-}\ell,q)\geq 2\myrho(p{-}\ell,q)$ holds for all
$\ell\leq p$. Thus, every local $\myrho(p{+}1{-}\ell,q)$-isomorphism is
also a local $\myrho(p{-}\ell,q)$-isomorphism. Furthermore, for $\ell\in[0,p{-}1]$,
an $i$-near move of Spoiler in round $\ell{+}1$ of the $(p,q)$-game
also is an $i$-near move for Spoiler in round $\ell{+}1$ of the
$(p{+}1,q)$-game, because
$\myrho(p{-}\ell,q)-\myrho(p{-}(\ell{+}1),q) \leq
\myrho(p{+}1{-}\ell,q)-\myrho(p{+}1{-}(\ell{+}1),q)$ (this easily
follows by noting that
$\myrho(p{+}1{-}(\ell{+}1),q)=\myrho(p{-}\ell,q)$ and using the fact
that $\myrho(p{+}1{-}\ell,q)\geq 2\myrho(p{-}\ell,q)$).

Next, we prove an \emph{Ehrenfeucht-\Fraisse\
  Theorem} for the $(p,q)$-game and  $\FOplus[p,q]$.

\begin{theorem}\label{thm:EF}
  Let $k,p,q\in\Nat$ such that $k+p\le q$. Then for all
  $\sigma$-structures $\CA,\CB$ and for all tuples $\vec a\in A^k,\vec
  b\in B^k$ the following are equivalent.
   \begin{enumerate}
  \item\label{item:EFThm:one} Duplicator has a winning strategy for the $(p,q)$-game on $(\CA,\vec a)$,
    $(\CB,\vec b)$.
  \item\label{item:EFThm:two} For all formulas $\phi(x_1,\ldots,x_k)\in\FOplus[p,q]$ we have 
      $\CA\models\phi[\vec a] \iff  \CB\models\phi[\vec b]$.
  \end{enumerate}
\end{theorem}

\begin{proof}
  Observe that by \eqref{eq:tpVsFormulas}, assertion \ref{item:EFThm:two} can be equivalently stated as $\tp_{p,q}(\CA,\vec a)=\tp_{p,q}(\CB,\vec b)$. We shall make use of this equivalence several times in the proof.

  Let $q\in\Nat$, and let $\CA,\CB$ be $\sigma$-structures.
  By induction on $p\le q$ we prove the equivalence ``\ref{item:EFThm:one}$\Longleftrightarrow$\ref{item:EFThm:two}'' for all $k\in\Nat$ such that $k\le q-p$ and all tuples $\vec a\in A^{k},\vec b\in B^{k}$.

  The base step $p=0$ follows immediately from the definitions.
  For the inductive step $p\to p{+}1$, let $p\in[0,q{-}1]$,
  consider an arbitrary $k\in\Nat$  such that $k\le q-(p{+}1)$, and
  consider tuples $\vec a\in A^k$, $\vec b\in B^k$.

  \medskip

  To prove 
  ``\ref{item:EFThm:two}$\Longrightarrow$\ref{item:EFThm:one}'' for $p{+}1$,
  assume that for all formulas $\phi(x_1,\ldots,x_k)\in
  \FOplus[p{+}1,q]$ we have $\CA\models\phi[\vec
  a]\iff\CB\models\phi[\vec b]$.
  We shall describe a winning strategy for Duplicator in the
  $(p{+}1,q)$-game on $(\CA,\vec a)$, $(\CB,\vec b)$.

  Suppose first that in the
  first round of the game, Spoiler makes an $i$-near move and, without
  loss of generality, picks an $a_{k+1}\in A$ such that
  $d\coloneqq \dist^\CA(a_i,a_{k+1})\le \myrho(p{+}1,q)-\myrho(p,q)$. Let
  $\theta\coloneqq\bigwedge\tp_{p,q}(\CA,(a_1,\ldots,a_{k+1}))$. Then
  $\theta(x_1,\ldots,x_{k+1})\in\FOplus[p,q]$ and hence $\exists x_{k+1}\big(\dist(x_i,x_{k+1})\le
  d\wedge\theta\big)\in\FOplus[p{+}1,q]$.
  Since
  \[
    \CA \ \models \ \exists x_{k+1}\big(\dist(x_i,x_{k+1})\le
    d\wedge\theta\big) \; [\vec a],
  \]
  we have $\CB\models\exists x_{k+1}\big(\dist(x_i,x_{k+1})\le
  d\wedge\theta\big) [\vec b]$. Hence there is a $b_{k+1}\in B$
  such that $\dist^\CB(b_i,b_{k+1})\le d$ and
  $\tp_{p,q}\big(\CB,(b_1,\ldots,b_{k+1})\big)=\tp_{p,q}\big(\CA,(a_1,\ldots,a_{k+1})\big)$. Duplicator
  answers by choosing $b_{k+1}$. By the induction hypothesis, 
  Duplicator has a winning strategy for the $(p,q)$-game on
  $(\CA,(a_1,\ldots,a_{k+1}))$,
  $(\CB,(b_1,\ldots,b_{k+1}))$. Following this strategy during the
  remaining $p$ rounds ensures that she wins the $(p{+}1,q)$-game on
  $(\CA,\vec a)$, $(\CB,\vec b)$. 

  Suppose next that in the first round of the $(p{+}1,q)$-game on
  $(\CA,\vec a)$, $(\CB,\vec b)$ Spoiler makes a
  far move and, without
  loss of generality, that he picks an $a_{k+1}\in A$ such that
  $\dist^\CA(a_i,a_{k+1})> \myrho(p{+}1,q)-\myrho(p,q)$ for all
  $i\in[k]$. We let
  $\theta\coloneqq\bigwedge\tp_{p,q}\big(\CA,(a_1,\ldots,a_{k+1})\big)$
  and observe that $\exists x_{k+1}\theta\in\FOplus[p{+}1,q]$. Since
  \[
    \CA \ \models \ \exists x_{k+1}\theta\; [\vec a],
  \]
  we have $\CB\models\exists x_{k+1}\theta\,[\vec b]$. Hence there is a $b_{k+1}\in B$
  such that 
  \[
    \tp_{p,q}(\CB,(b_1,\ldots,b_{k+1})) \ = \ \tp_{p,q}(\CA,(a_1,\ldots,a_{k+1})).\]
  Duplicator
  answers by choosing such a $b_{k+1}$. By the induction hypothesis,
  Duplicator has a winning strategy for the $(p,q)$-game on
  $(\CA,(a_1,\ldots,a_{k+1}))$, $(\CB,(b_1,\ldots,b_{k+1}))$. Following this strategy during the
  remaining $p$ rounds ensures that she wins the $(p{+}1,q)$-game on
  $(\CA,\vec a)$, $(\CB,\vec b)$.

  \medskip
  
  To prove the implication ``\ref{item:EFThm:one}$\Longrightarrow$\ref{item:EFThm:two}'' for $p{+}1$, assume that Duplicator
  has a winning strategy for the $(p{+}1,q)$-game on $(\CA,\vec a)$,
  $(\CB,\vec b)$. By the definition of $\FOplus[p{+}1,q]$ it suffices to prove \ref{item:EFThm:two} for all formulas
  $\phi(x_1,\ldots,x_k)\in\FOplus[p{+}1,q]$ that are (a) in
  $\FOplus[p,q]$, or (b) distance atoms $\dist(x_i,x_j)\le d$ with
  $d\le \myrho(p{+}1,q)$ and $i,j\in[k]$, or (c) formulas $\exists x_{k+1}\,\psi$
  where $\psi(x_1,\ldots,x_{k+1})\in\FOplus[p,q]$, 
  or (d) formulas
  \begin{equation}
    \label{eq:2}
    \exists x_{k+1}\,\big(\dist(x_i,x_{k+1})\le
    d\ \wedge \psi\big),
  \end{equation}
  where $i\in[k]$, $d\le \myrho(p{+}1,q)-\myrho(p,q)$, and
  $\psi(x_1,\ldots,x_{k+1})\in\FOplus[p,q]$.

  For formulas of type (a), the assertion holds by the induction
  hypothesis (here we use that Duplicator having a winning strategy
  for the $(p{+}1,q)$-game on $(\CA,\vec a)$, $(\CB,\vec b)$ implies
  that she also has a winning strategy for the $(p,q)$-game on
  $(\CA,\vec a)$, $(\CB,\vec b)$).

  For formulas of type (b), the assertion follows from the fact that
  the mapping $a_j\mapsto b_j$ for $j\in[k]$ is a local
  $\myrho(p{+}1,q)$-isomorphism.

  For formulas of type (d), suppose that $\phi(x_1,\ldots,x_k)$ is the
  formula in \eqref{eq:2}. By symmetry, we may assume without loss of
  generality that $\CA\models\phi[a_1,\ldots,a_k]$. Let $a_{k+1}\in A$
  such that $\dist^\CA(a_i,a_{k+1})\le d\le \myrho(p{+}1,q)-\myrho(p,q)$ and
  $\CA\models\psi[a_1,\ldots,a_{k+1}]$. Let $b_{k+1}\in B$ be
  Duplicator's answer if Spoiler makes an $i$-near move and picks
  $a_{k+1}$ in the first round of the $(p{+}1,q)$-game on $(\CA,\vec
  a)$, $(\CB,\vec b)$. Since Duplicator plays according to her winning
  strategy, we know that
  $\dist^{\CB}(b_i,b_{k+1})\le d$ and that Duplicator has a
  winning strategy for the $(p,q)$-game on
  $(\CA,(a_1,\ldots,a_{k+1}))$, $(\CB,(b_1,\ldots,b_{k+1}))$. By the
  induction hypothesis and the assumption 
  $\CA\models\psi[a_1,\ldots,a_{k+1}]$, this implies
  $\CB\models\psi[b_1,\ldots,b_{k+1}]$ and thus
  $\CB\models\phi[b_1,\ldots,b_{k}]$.

  Finally, for formulas of type (c), let
  $\phi(x_1,\ldots,x_k)$ be of the form $\exists x_{k+1}\,\psi$ where $\psi(x_1,\ldots,x_{k+1})\in\FOplus[p,q]$. Without loss of
  generality we assume that $\CA\models\phi[a_1,\ldots,a_k]$. Let $a_{k+1}\in A$
  such that 
  $\CA\models\psi[a_1,\ldots,a_{k+1}]$.
  If $\dist^\CA(a_i,a_{k+1})\le \myrho(p{+}1,q)-\myrho(p,q)$ for some $i\in[k]$, we argue as for
  formulas of type (d). Otherwise, let $b_{k+1}\in B$ be
  Duplicator's answer if Spoiler makes a far move and picks
  $a_{k+1}$ in the first round of the $(p{+}1,q)$-game on $(\CA,\vec
  a)$, $(\CB,\vec b)$.
  Since Duplicator plays according to her winning
  strategy, we know that she has  
  winning strategy for the $(p,q)$-game on
  $(\CA,(a_1,\ldots,a_{k+1}))$, $(\CB,(b_1,\ldots,b_{k+1}))$. By the
  induction hypothesis and the assumption 
  $\CA\models\psi[a_1,\ldots,a_{k+1}]$, this implies
  $\CB\models\psi[b_1,\ldots,b_{k+1}]$ and thus
  $\CB\models\phi[b_1,\ldots,b_{k}]$.
  
  This completes the proof of Theorem~\ref{thm:EF}.
\end{proof}

\section{Syntactically Local Formulas in $\FOplus[p,q]$}
\label{sec:local}
For all $p,q\in\Nat$ with
$p\le q$ we define a subclass
$\LocFOplus[p,q]\subseteq\FOplus[p,q]$ of \emph{syntactically
  local} formulas. The inductive definition of
$\LocFOplus[p,q]$ follows that of $\FOplus[p,q]$ except that
in the inductive step, no unrestricted quantification in formulas
$\exists y\,\psi$ is permitted. The precise definition is as follows.

\begin{definition}\label{def:LocFOpluspq}
  For $p,q\in\NN$ with $p\leq q$, we define the following sets
  $\LocFOplus[p,q]$:
  \begin{itemize}
    \item
      $\LocFOplus[0,q]\coloneqq\FOplus[0,q]$.
    \item
      For $p>0$, $\LocFOplus[p,q]$ is the set of
all formulas $\lambda$ with $|\free(\lambda)|\le q{-}p$ such that $\lambda$ is a
Boolean combination of formulas
in $\LocFOplus[p{-}1,q]$, distance atoms $\dist(x,y)\le d$ with
$d\le \myrho(p,q)$, and formulas
\begin{equation}\label{eq:locFO}
  \exists y\;\big(\dist(x,y)\le
  d\ \wedge\mu\big),
\end{equation}
where $\mu\in\LocFOplus[p{-}1,q]$ and $d\le
\myrho(p,q)-\myrho(p{-}1,q)$.
\uend
\end{itemize}
\end{definition}

We will usually denote formulas in $\LocFOplus[p,q]$ by
$\lambda,\mu,\nu$ to distinguish them from arbitrary formulas in $\FOplus[p,q]$, which we typically denote by $\phi,\psi,\chi$.

Observe that
the only \emph{sentences} (i.e., formulas with no free
variable) in $\LocFOplus[0,q]$ are Boolean combinations of formulas of
the form $R()$, where $R\in\sigma$ is a relation symbol of arity 0.
Whenever we introduce a quantifier in a formula of the form
\eqref{eq:locFO}, the variable $x$ remains free in the formula.
Thus, for all $p\leq q$, the only \emph{sentences} in
$\LocFOplus[p,q]$ are the sentences that are already present in
$\LocFOplus[0,q]$.
In particular, as for all $\lambda\in\LocFOplus[p,q]$ it holds that $|\free(\lambda)|+p\le
q$, this implies that $\LocFOplus[q,q]\subseteq\LocFOplus[0,q]$ for all $q\in\NN$.

The following lemma shows that all formulas in $\LocFOplus[p,q]$ are
$r$-local for $r\deff\myrho(p,q)$.

\begin{lemma}\label{lemma:locality}
  Let $k,p,q\in\Nat$ with $k\ge 1$ and $p\le q$. Then every $\lambda(z_1,\ldots,z_k)\in\LocFOplus[p,q]$ is
  $\myrho(p,q)$-local. 
\end{lemma}

\begin{proof}
  Let $q\in\NNpos$.
  
  Above,
  we have already seen that every \emph{sentence} in $\LocFOplus[p,q]$
  is, in fact, a Boolean combination of
  atoms of the form $R()$, where $R\in\sigma$ is a relation symbol of
  arity 0. Such sentences are $0$-local around every non-empty tuple
  $\vec z$ of variables (cf.\
 the definition provided at the beginning of
Section~\ref{sec:Gaifman}). Since $\myrho(p,q)\geq 0$,
  they also are $\myrho(p,q)$-local, for all $p\leq q$.

  Therefore, in the following we can restrict attention to formulas $\lambda$ in
  $\LocFOplus[p,q]$ that have at least one free variable. It is straightforward
  to see that it suffices to consider numbers $p<q$ and variable
  tuples $\vec z=(z_1,\ldots,z_k)$ such that $\free(\lambda)=\set{z_1,\ldots,z_k}$.
  
  For every $p<q$ let $r_p\deff \myrho(p,q)$. 
  By induction on $p$ we show for all $p\in[0,q{-}1]$ that every
  formula $\lambda\in\LocFOplus[p,q]$ with $|\free(\lambda)|=k\geq 1$
  is $r_p$-local around its free variables, i.e.,
  for all $\sigma$-structures $\CA$ and all tuples $\vec a\in A^k$ we have
  \[
    \CA\models\lambda[\vec a] \ \ \iff \ \ \CN_{r_p}^{\CA}(\vec
    a)\models\lambda[\vec a].
  \]
  In the following, whenever we consider a formula
  $\lambda\in\LocFOplus[p,q]$, we let $k\deff
\free(\lambda)$ and $\set{z_1,\ldots,z_k}=\free(\lambda)$, and we let
$\vec z=(z_1,\ldots,z_k)$. Note that $1\leq k\leq q{-}p$.
 \medskip

  For the base step $p=0$,
  consider a $\lambda\in\LocFOplus[0,q]$. Then, $\lambda$ is a Boolean
  combination of atomic $\FO$ formulas with variables among
  $z_1,\ldots,z_k$ and of distance atoms $\dist(z_i,z_j)\leq d$ with
  $i,j\in[k]$ and
  $d\leq\myrho(0,q)$, and $r_0=\myrho(0,q)\geq 1$ (according to condition \ref{myrho:constraint:zero} from the
  beginning of Section~\ref{sec:FOpluspq}).

  Let $\CA$ be a $\sigma$-structure and $\vec a=(a_1,\ldots,a_k)\in
  A^k$.
  Obviously, for every atomic $\FO$ formula $\mu$ with free variables
  among $z_1,\ldots,z_k$ we have $\CA\models\mu[\vec a]$ $\iff$
  $\CN_{r_0}^{\CA}(\vec a)\models\mu[\vec a]$.
  Concerning the distance atoms that occur in $\lambda$, note that if
  $\dist^\CA(a_i,a_j)\le r_0$, then all vertices $a$ on a shortest path from
  $a_i$ to $a_j$ in the Gaifman graph $G_{\CA}$ are in $N_{r_0}^\CA(a_i)$ and 
  $\dist^\CA(a_i,a_j)=\dist^{\CN_{r_0}(\vec a)}(a_i,a_j)$.
  This proves that $\CA\models\lambda[\vec a]$ $\iff$ $\CN_{r_0}^{\CA}(\vec
    a)\models\lambda[\vec a]$ and completes the base case $p=0$.
\medskip

  For the inductive step from $p{-}1$ to $p$, suppose that the assertion is proved for $p{-}1$
  and let $\lambda\in \LocFOplus[p,q]$. Without loss of
  generality, we may assume that (a) $\lambda\in\LocFOplus[p{-}1,q]$, or
  (b) $\lambda$ is of the form $\dist(z_i,z_j)\le d$ for some
  $d\le \myrho(p,q)$ and for $i,j\in[k]$, or (c) $\lambda$ is of the form
  \begin{equation}
    \label{eq:3}
      \exists y\;\big(\dist(z_i,y)\le d\ \wedge\mu\big),
  \end{equation}
  where $i\in[k]$, $d\le \myrho(p,q)-\myrho(p{-}1,q)$, and
  $\mu\in\LocFOplus[p{-}1,q]$ with $\free(\mu)\subseteq\set{z_1,\ldots,z_k,y}$.

  For formulas of type (a), we apply the induction hypothesis (here we
  use the fact that condition \ref{myrho:constraint:inc} from the
  beginning of Section~\ref{sec:FOpluspq} implies that $r_p\geq r_{p-1}$), and for
  formulas of type (b) we argue as in the base step. So assume
  that $\lambda$ is the formula in \eqref{eq:3} of type
  (c). Let $\CA$ be a $\sigma$-structure and let $\vec
  a=(a_1,\ldots,a_k)\in A^k$. We need to show that
  $\CA\models\lambda[\vec a]  \iff \CN_{r_p}^{\CA}(\vec
    a)\models\lambda[\vec a]$.
  
  For the forward direction, suppose that $\CA\models\lambda[\vec
  a]$. Let $b\in A$ such that
  $d'\coloneqq\dist^{\CA}(a_i,b)\le d \leq r_{p}-r_{p-1}$ and
  $\CA\models\mu[\vec a']$, where $\vec
  a'\coloneqq(a_1,\ldots,a_k,b)$. Let $\CA'\coloneqq\CN^\CA_{r_{p}}(\vec a)$ and
 $\CA''\coloneqq\CN^\CA_{r_{p-1}}(\vec a')$. 
Since 
  $N^\CA_{r_{p-1}}(\vec a')\subseteq N^\CA_{r_{p}}(\vec a)$, we have
 $\CA''=\CN^{\CA'}_{r_{p-1}}(\vec a')$.

 By the induction hypothesis applied to $\CA$ we have
 $\CA''\models\mu[\vec a']$. Thus, by the induction hypothesis applied to $\CA'$ we have
 $\CA'\models\mu[\vec a']$. As
 $d'=\dist^\CA(a_i,b)=\dist^{\CA'}(a_i,b)$,
 we thus have $\CA'\models\lambda[\vec a]$.

 The backward direction can be proved similarly:
 suppose that $\CA'\models\lambda[\vec
 a]$ for $\CA'\coloneqq\CN^\CA_{r_{p}}(\vec a)$.
 Then, there exists a $b\in N^\CA_{r_{p}}(\vec a)$ such that
  $d'\coloneqq\dist^{\CA'}(a_i,b)\le d \leq r_{p}-r_{p-1}$ and
  $\CA'\models\mu[\vec a']$, where $\vec
  a'\coloneqq(a_1,\ldots,a_k,b)$.
  By the induction hypothesis applied to $\CA'$, we obtain that
  $\CA''\models\mu[\vec a']$, where $\CA''\coloneqq
  \CN^{\CA'}_{r_{p-1}}(\vec a')$.
  Since  $N^{\CA'}_{r_{p-1}}(\vec a')\subseteq  N^{\CA}_{r_{p}}(\vec a)$, we have
 $\CA''=\CN^\CA_{r_{p-1}}(\vec a')$.
  By the induction hypothesis applied to $\CA$, we obtain that
  $\CA\models\mu[\vec a']$.
  As
  $d'=\dist^{\CA'}(a_i,b)=\dist^{\CA}(a_i,b)$,
 we thus have $\CA\models\lambda[\vec a]$. This completes the
 inductive step.
 In summary, the proof of Lemma~\ref{lemma:locality} is complete.
\end{proof}

Regarding normalisation, we just inherit the definitions from
$\FOplus[p,q]$: we define $\normLocFOplus[p,q]$ to be the set of
normalised formulas $\tilde\phi$, for all $\phi\in\LocFOplus[p,q]$ with
$\free(\phi)\subseteq\{x_1,\ldots,x_{q-p}\}$. Again, it follows from Remark~\ref{rem:norm-local} that $\normLocFOplus[p,q]\subseteq\LocFOplus[p,q]$.

\subsection{Local Types}\label{sec:localTypes}

Let $k,p,q\in\Nat$ with $k+p\leq q$ and $k\geq 1$.
A \emph{local $(k,p,q)$-type} 
is a set $\Lambda$ of formulas $\mu(x_1,\ldots,x_k)\in\normLocFOplus[p,q]$ 
such that for each
$\lambda(x_1,\ldots,x_k)\in\LocFOplus[p,q]$, either $\normalised{\lambda}\in
\Lambda$ or $\normalised{\neg\lambda}\in\Lambda$, but not both. Since
$\normLocFOplus[p,q]$ is finite,
it is straightforward to see that
there is a formula $\lambda\in
\normLocFOplus[p,q]$ such that for all $\sigma$-structures
$\CA$ and tuples $\vec a\in A^k$ we have
\[
  \CA\models\lambda[\vec a]\ \ \iff \ \  \CA\models\mu[\vec a]\text{ for all
  }\mu\in\Lambda.
\]
Indeed, we can take $\lambda$ to be the normalisation of the conjunction
of all formulas in $\Lambda$, and we will henceforth denote $\lambda$ as
$\bigwedge\Lambda$.

For a $\sigma$-structure $\CA$ and a tuple $\vec a=(a_1,\ldots,a_k)\in
A^k$, the \emph{local $(p,q)$-type} of $\vec a$ in $\CA$ is the set
$\ltp_{p,q}(\CA,\vec a)$ of
all formulas $\lambda(x_1,\ldots,x_k)\in\normLocFOplus[p,q]$ such that
$\CA\models\lambda[\vec a]$. Since all
$\lambda(x_1,\ldots,x_k)\in\LocFOplus[p,q]$ are $\myrho(p,q)$-local
by Lemma~\ref{lemma:locality}, for all induced substructures
$\CA'$ of $\CA$ with $N^{\CA}_{\myrho(p,q)}(\vec a)\subseteq A'$ we have 
\[
  \ltp_{p,q}(\CA,\vec a)\ \ = \ \ \ltp_{p,q}(\CA',\vec a).
\]
In particular, this holds for $\CA'\coloneqq\CN^{\CA}_{\myrho(p,q)}(\vec a)$.
Moreover, for all $\sigma$-structures
$\CB$ and tuples $\vec b\in B^k$ we have
\begin{equation}\label{eq:localType}
\begin{array}{cl}
& \ltp_{p,q}(\CA,\vec a)=\ltp_{p,q}(\CB,\vec b)
                    \medskip\\
  \iff
  &  \text{for all } \lambda(x_1,\ldots,x_k)\in\LocFOplus[p,q] : \
    \big(\,
      \CA\models\lambda[\vec a]\iff\CB\models\lambda[\vec b]
    \,\big).
\end{array}  
\end{equation}
Also note that for $\lambda\coloneqq\bigwedge\ltp_{p,q}(\CA,\vec a)$ we have
\[
  \CB\models\lambda[\vec b] \ \ \iff \ \ \ltp_{p,q}(\CA,\vec
  a)=\ltp_{p,q}(\CB,\vec b).
\]

\subsection{Independence Sentences}\label{subsec:IndependenceSentences}
In \cite{GroheKS17,GroheS18}, an \emph{$(r,s)$-independence sentence} is defined to be an
$\FOplus$ sentence of the form
\[
  \exists x_1\cdots \exists x_{s'}\left(\bigwedge_{1\le i<j\le
      s'}\dist(x_i,x_j)>r'\wedge \bigwedge_{1\le i\le s'}\alpha(x_i)\right),
\]
where $r'\le r,s'\le s$, and $\alpha(x)$ is a quantifier-free
formula. Compare this to the basic local sentences in Gaifman's normal
form theorem (Theorem~\ref{thm:gaifman}). Those sentences have similar
structure, but instead of a quantifier free formula $\alpha(x)$ there
is an $\frac{r'}{2}$-local formula.
The following version of independence sentences generalises both
notions.

For numbers $p,q,r,s\geq 0$, a \emph{$(p,q,r,s)$-independence
  sentence} is an $\FOplus$ sentence of the form
\begin{equation}\label{eq:indepSentence}
\exists x_1\cdots\exists x_{s'}\left(\bigwedge_{1\le i<j\le s'}\dist(x_i,x_j)>2r'\wedge\bigwedge_{1\le i\le s'}\lambda(x_i)\right),
\end{equation}
where $p'\le p$, $q'\le q$, $\myrho(p',q')\le r'\le r$, $s'\le s$, and
$\lambda(x)\in\LocFOplus[p',q']$. Note that such a sentence is a basic
local sentence (cf.\ Section~\ref{sec:Gaifman}) because the
formula $\lambda(x)$ is $\myrho(p',q')$-local by
Lemma~\ref{lemma:locality}, and hence it is $r'$-local.

We let $\IS[p,q,r,s]$ be the set of all
$(p,q,r,s)$-independence sentences, and we let
$\normIS[p,q,r,s]$ be the set of all \emph{normalised} $(p,q,r,s)$-independence
sentences, i.e., sentences of the form \eqref{eq:indepSentence} where
$\lambda\in\normLocFOplus[p',q']$.
Note that $\normIS[p,q,r,s]$
is finite because $\normLocFOplus[p,q]$ is finite.
A (normalised) \emph{$(p,q,r,s)$-independence literal} is either a (normalised)
$(p,q,r,s)$-independence sentence or a  negated (normalised) $(p,q,r,s)$-independence
sentence.

A \emph{$(p,q,r,s)$-independence type} 
is a set $\Xi$ of normalised $(p,q,r,s)$-independence literals 
such that for each sentence $\xi\in\normIS[p,q,r,s]$, either $\xi\in\Xi$ or $\neg\xi\in\Xi$.
For a $(p,q,r,s)$-independence type $\Xi$, by $\bigwedge\Xi$ we denote the
conjunction of all literals in $\Xi$.

For a $\sigma$-structure $\CA$, the \emph{$(p,q,r,s)$-independence type} of $\CA$ is
the set $\itp_{p,q,r,s}(\CA)$ of all normalised
$(p,q,r,s)$-independence literals satisfied by $\CA$.

Note that for all $\sigma$-structures $\CA,\CB$ we have
\begin{equation}\label{eq:itp}
  \begin{array}{cl}
& \itp_{p,q,r,s}(\CA) = \itp_{p,q,r,s}(\CB)
\smallskip\\
    \iff
    & \text{for all } \xi\in\IS[p,q,r,s]\ : \ \
      \big( \CA\models\xi \ \iff \ \CB\models\xi \big),
  \end{array}
\end{equation}
and
\[
  \CB\models\bigwedge\itp_{p,q,r,s}(\CA) \quad \iff \quad \itp_{p,q,r,s}(\CB)=\itp_{p,q,r,s}(\CA).
\]

\section{A Gaifman Normal Form Theorem for $\FOplus[p,q]$}
\label{sec:main}

For $m,r\in\Nat$, a $\sigma$-structure $\CA$, and a tuple $\vec
a=(a_1,\ldots,a_m)\in A^m$ we let
$G^{\CA}_{\vec a,r}$ be the undirected graph with vertex set
$[m]$ and an edge between vertices $i,j$ if and only if
$i\neq j$ and $\dist^{\CA}(a_i,a_j)\leq r$.
An \emph{$r$-component of $\vec a$} is the vertex set of a connected
component of the graph $G^{\CA}_{\vec a,r}$. Note that for $m=0$,
$G^{\CA}_{\vec a,r}$ is the empty graph, which by definition has no
connected components.

We now are ready to state and prove our main technical lemma.

\begin{lemma}\label{lem:main}
  Let $k,p,q\in\Nat$ with $k+p\leq q$, and let
  $r\coloneqq \myrho(p,q)$.
  Let $\CA,\CB$ be
  $\sigma$-structures and $\vec a=(a_1,\ldots,a_k)\in A^k$, $\vec b=(b_1,\ldots,b_k)\in B^k$ such that the following conditions are satisfied.
  \begin{eroman}
  \item\label{item:i:lem:main}
    $G^{\CA}_{\vec a,r}=G^{\CB}_{\vec b,r}$ and for every
    $r$-component $J$ of $\vec a$ we have
    $\ltp_{p,q}(\CA,\vec a_J) = \ltp_{p,q}(\CB,\vec b_J)$.
  
  \item\label{item:ii:lem:main}
   $\itp_{p-1,q,\rHalbe,k+p}(\CA) = \itp_{p-1,q,\rHalbe,k+p}(\CB)$.
  \end{eroman}
  Then
  $\tp_{p,q}(\CA,\vec a) = \tp_{p,q}(\CB,\vec b)$.
\end{lemma}

Note that condition \ref{item:i:lem:main} is trivially satisfied for $k=0$.

\begin{proof}[Proof of Lemma~\ref{lem:main}]
  For $\ell\in[0,p]$ let $r_\ell \coloneqq
  \myrho(p{-}\ell,q)$. Then, $r=r_0 \geq  r_1 \geq \cdots \geq 
  r_p=\myrho(0,q)\geq 1$ and we have $r_\ell\geq 2r_{\ell+1}$
  for all $\ell\in[0,p{-}1]$ (this holds because $\myrho$ satisfies the
  conditions \ref{myrho:constraint:zero} and
  \ref{myrho:constraint:inc} formulated at the beginning of Section~\ref{sec:FOpluspq}). 
  We shall prove that Duplicator has a winning strategy for the
  $(p,q)$-game on $(\CA,\vec a),(\CB,\vec b)$. Then the assertion follows
  from Theorem~\ref{thm:EF} and equation \eqref{eq:tpVsFormulas}.

  For $\ell\in[0,p]$, the position of the $(p,q)$-game on $(\CA,\vec
  a)$, $(\CB,\vec b)$ 
  after $\ell$ rounds will be $\vec a^\ell,\vec b^\ell$, where $\vec
  a^\ell=(a_1,\ldots,a_{k+\ell})\in A^{k+\ell}$ and $\vec
  b^\ell=(b_1,\ldots,b_{k+\ell})\in B^{k+\ell}$.

  We shall prove, by induction on $\ell$, that Duplicator can maintain
  the following invariants
  for all $\ell\in [0,p]$.

  \begin{enumerate}
  \item[(I)] For every $r_\ell$-component $J$ of $\vec a^\ell$
    we have
    $\ltp_{p-\ell,q}(\CA,\vec a^{\ell}_J)=\ltp_{p-\ell,q}(\CB,\vec b^{\ell}_J)$.
  \item[(II)] For all $i,j\in[k{+}\ell]$, either
    $\dist^\CA(a_i,a_j)=\dist^\CB(b_i,b_j)\leq r_\ell$    
    or $\dist^\CA(a_i,a_j)>r_\ell$
    and  $\dist^\CB(b_i,b_j)>r_\ell$.
  \end{enumerate}
  Observe that (II) implies that $G^{\CA}_{\vec a^\ell,r_\ell}=G^{\CB}_{\vec b^\ell,r_\ell}$ and hence that the tuples $\vec a^\ell$ and $\vec
  b^\ell$ have the same $r_\ell$-components.

  It is straightforward to
  verify that if (I) and (II) hold for all $\ell\in [0,p]$,
  then Duplicator's winning condition is satisfied, i.e., for each
  $\ell\in [0,p]$ the mapping $a_j\mapsto b_j$ for $j\in [k{+}\ell]$
  is a local $\myrho(p{-}\ell,q)$-isomorphism:
  from (II) and the fact that $r_\ell=\myrho(p{-}\ell,q)$, we obtain that all the
  distance conditions are satisfied. From (I) we obtain that the
    usual ``local isomorphism'' condition is satisfied for tuples
    built from indices in the same $r_p$-component of $\vec
    a^{p}$. And we know that $r_p=\myrho(0,q)\geq 1$ --- hence, tuples
    containing indices from different $r_p$-components of $\vec a^p$
    do not belong to any relation present in $\CA,\CB$.
  
  \medskip\noindent
  \textit{Base step: $\ell=0$.}\\
  If $k=0$, then (I) and (II) are trivially satisfied. If $k\geq 1$,
  then (I) immediately follows from assumption~\ref{item:i:lem:main}
  since $r_0=r$ and $\vec a^0=\vec a$,
  and (II) follows from assumption~\ref{item:i:lem:main} by using
  equation \eqref{eq:localType} and noting that $\LocFOplus[p,q]$
  contains all distance atoms $\dist(x_i,x_j)\leq d$ and their negations
  for $d\leq \myrho(p,q) = r_0$ and $i,j\in[k]$.
  
  \medskip\noindent
  \textit{Inductive step: $\ell\to \ell{+}1$.} 
  Let $\ell\in [0,p{-}1]$.\\
    Assume that (I) and (II) hold for the positions
  $\vec a^\ell=(a_1,\ldots,a_{k+\ell}),\vec b^\ell=(b_1,\ldots,b_{k+\ell})$
  reached after $\ell$ rounds.
  
  It will be convenient to let $\vec x=(x_1,\ldots,x_{k+\ell})$ and $\vec
  x'=(x_1,\ldots,x_{k+\ell},x_{k+\ell+1})$. 
  
  \medskip\noindent
  \textit{Case 1: }In round $\ell{+}1$ of the game, Spoiler decides to
  make an $i$-near move
  for some $i\in [k{+}\ell]$.\\
  By symmetry, we may assume without loss
  of generality that he picks an element
  $a_{k+\ell+1}\in A$ such that
$d\coloneqq\dist^{\CA}(a_i,a_{k+\ell+1})\le
 \myrho(p{-}\ell,q)-\myrho(p{-}(\ell{+}1),q)=
    r_{\ell}-r_{\ell+1}$. Let $\vec
    a^{\ell+1}\coloneqq(a_1,\ldots,a_{k+\ell},a_{k+\ell+1})$.
 
   \medskip\noindent \textit{Case 1a:}
   $\dist^\CA(a_j,a_{k+\ell+1})\le r_{\ell+1}$ for some
   $j\in[k{+}\ell]$.\\
   Pick an arbitrary $j\in[k{+}\ell]$ such that $\dist^\CA(a_j,a_{k+\ell+1})\le r_{\ell+1}$.
   Let $I'$ be the $r_{\ell+1}$-component of $\vec a^{\ell+1}$ that
   contains $k{+}\ell{+}1$. Then $j\in I'$ by
   our choice of $j$.
   Let $I$ be the $r_{\ell}$-component of $\vec a^\ell$ that contains
   $j$ (in particular, $I\subseteq[k{+}\ell]$).
   Then $I'\setminus\{k{+}\ell{+}1\}\subseteq I$, because
   $r_{\ell+1}\le r_\ell$.
   Furthermore, $i\in I$ because
   \[
     \dist^\CA(a_i,a_j)\ \le \
     \dist^\CA(a_i,a_{k+\ell+1})+\dist^\CA(a_j,a_{k+\ell+1}) \ \le \
     (r_{\ell}-r_{\ell+1})+r_{\ell+1} \ = \ r_\ell.
   \]
   Let $\psi(\vec
   x'_{I'})\coloneqq\bigwedge\ltp_{p-\ell-1,q}(\CA,\vec
   a^{\ell+1}_{I'})$. Then $\psi(\vec
   x'_{I'})\in\LocFOplus[p{-}\ell{-}1,q]$.
   Let
   \[
     \lambda(\vec x_I)\ \ \coloneqq \ \ \exists
     x_{k+\ell+1}\Big(\dist(x_i,x_{k+\ell+1})\le
       d\wedge\Big(\bigwedge_{j'\in I\setminus I'}\dist(x_{j'},x_{k+\ell+1})>r_{\ell+1}\wedge\psi(\vec
     x'_{I'})\Big)\Big).
   \]
   Then $\lambda\in\LocFOplus[p{-}\ell,q]$.
   Since $\CA\models\lambda[\vec a^\ell_I]$, we have
   $\CB\models\lambda[\vec b^\ell_I]$ by the induction hypothesis
   (I) and equation \eqref{eq:localType}.
   Hence, there is a
   $b_{k+\ell+1}\in B$ such that $\dist^\CB(b_i,b_{k+\ell+1})\le d$ and
   $\dist^{\CB}(b_{j'},b_{k+\ell+1})>r_{\ell+1}$ for all $j'\in
   I\setminus I'$ and
   $\CB\models \psi[\vec b^{\ell+1}_{I'}]$, where 
   $\vec b^{\ell+1}\coloneqq(b_1,\ldots,b_{k+\ell},b_{k+\ell+1})$. This element
   $b_{k+\ell+1}$ is Duplicator's answer in round $\ell{+}1$.

   To establish invariant  (I) for the
   new position $\vec a^{\ell+1},\vec b^{\ell+1}$, let $J'$
   be an $r_{\ell+1}$-component of $\vec a^{\ell+1}$.
   If $J'=I'$ then $\ltp_{p-\ell-1,q}(\CA,\vec
   a^{\ell+1}_{J'})=\ltp_{p-\ell-1,q}(\CB,\vec
   b^{\ell+1}_{J'})$ because $\CB\models \psi[\vec
   b^{\ell+1}_{J'}]$.
   So assume that $J'\neq I'$. Then $J'\subseteq [k{+}\ell]$, and thus
   there is an $r_\ell$-component $J$ of $\vec a^\ell$ such that
   $J'\subseteq J$. Then the induction hypothesis $\ltp_{p-\ell,q}(\CA,\vec
   a^\ell_J)=\ltp_{p-\ell,q}(\CB,\vec b^\ell_J)$ implies
   $\ltp_{p-\ell-1,q}(\CA,\vec
   a^{\ell+1}_{J'})=\ltp_{p-\ell-1,q}(\CB,\vec b^{\ell+1}_{J'})$.

   It remains to establish (II). Since $r_{\ell+1}\le r_\ell$ and by the
   induction hypothesis, it suffices to prove that for all
   $j'\in[k{+}\ell]$, either
   $\dist^\CA(a_{j'},a_{k+\ell+1})=\dist^\CB(b_{j'},b_{k+\ell+1})\leq r_{\ell+1}$ 
   or
   $\dist^\CA(a_{j'},a_{k+\ell+1})>r_{\ell+1}$ and
   $\dist^\CB(b_{j'},b_{k+\ell+1})>r_{\ell+1}$. 
   Consider an arbitrary $j'\in[k{+}\ell]$.
   
   Assume first that
   $\dist^\CA(a_{j'},a_{k+\ell+1})\le r_{\ell+1}$. 
   Then $j'\in I'$, and
   $\ltp_{p-\ell-1,q}(\CA,\vec
   a^{\ell+1}_{I'})=\ltp_{p-\ell-1,q}(\CB,\vec b^{\ell+1}_{I'})$
   implies
   $\dist^\CA(a_{j'},a_{k+\ell+1})=\dist^\CB(b_{j'},b_{k+\ell+1})$
   because for all $d'\leq r_{\ell+1}$ the type $\ltp_{p-\ell-1,q}(\CA,\vec
   a^{\ell+1}_{I'})$  either contains a formula
   expressing $\dist(x_{j'},x_{k+\ell+1})\leq d'$, or it contains a
   formula expressing that $\dist(x_{j'},x_{k+\ell+1})> d'$.

   Now
   assume $\dist^\CA(a_{j'},a_{k+\ell+1})> r_{\ell+1}$. 
   If $j'\in I'$,
   then $\dist^\CB(b_{j'},b_{k+\ell+1})>r_{\ell+1}$ because $\ltp_{p-\ell-1,q}(\CA,\vec
   a^{\ell+1}_{I'})=\ltp_{p-\ell-1,q}(\CB,\vec b^{\ell+1}_{I'})$
   and $\ltp_{p-\ell-1,q}(\CA,\vec
   a^{\ell+1}_{I'})$ contains a formula expressing that $\dist(x_{j'},x_{k+\ell+1})> r_{\ell+1}$. 
   If
   $j'\in I\setminus I'$, then
   $\dist^\CB(b_{j'},b_{k+\ell+1})>r_{\ell+1}$ by the choice of
   $b_{k+\ell+1}$
   according to the formula $\lambda(\vec x_I)$.
   If $j'\not\in I$ then
   $j'$ and $j$ belong to different connected components of
   $G^{\CA}_{\vec a^\ell,r_\ell}$; in particular, this graph does not
   contain the edge $\{j,j'\}$, i.e.,
   $\dist^{\CA}(a_j,a_{j'})>r_\ell$. By the induction hypothesis we
   obtain that $\dist^{\CB}(b_j,b_{j'})>r_\ell$.
   By definition, $r_\ell = \myrho(p{-}\ell,q)\geq
   2\myrho(p{-}\ell{-}1,q)=2r_{\ell+1}$ (this holds because $\myrho$
   satisfies condition \ref{myrho:constraint:inc} formulated at the beginning of Section~\ref{sec:FOpluspq}).
  Since $j\in I'$, we already know that $\dist^{\CB}(b_j,b_{k+\ell+1})=\dist^{\CA}(a_j,a_{k+\ell+1})\leq r_{\ell+1}$.
   It follows that
   $\dist^\CB(b_{j'},b_{k+\ell+1})>r_{\ell+1}$, because otherwise we
   would have $\dist^{\CB}(b_j,b_{j'})\leq
   \dist^{\CB}(b_j,b_{k+\ell+1}) +
   \dist^{\CB}(b_{j'},b_{k+\ell+1})\leq 2r_{\ell+1} \leq r_{\ell}$.  

   \medskip\noindent
   \textit{Case 1b:} $\dist^\CA(a_j,a_{k+\ell+1})> r_{\ell+1}$ for all
   $j\in[k{+}\ell]$.\\
   Then 
   $k{+}\ell{+}1$ is an isolated node in the graph $G^{\CA}_{\vec a^{\ell+1},r_{\ell+1}}$ and thus 
   the $r_{\ell+1}$-component of $\vec a^{\ell+1}$ that
   contains $k{+}\ell{+}1$ is $\{k{+}\ell{+}1\}$.
   Let $I$ be the $r_\ell$-component of $\vec
   a^\ell$ that contains $i$.

   Let $\psi(x_{k+\ell+1})\coloneqq\bigwedge\ltp_{p-\ell-1,q}(\CA,a_{k+\ell+1})$ 
   and
   \[
     \lambda(\vec x_I) \ \ \coloneqq \ \ \exists
     x_{k+\ell+1}\left(\dist(x_i,x_{k+\ell+1})\le d\wedge\left(\bigwedge_{j\in I}\dist(x_j,x_{k+\ell+1})>r_{\ell+1}\wedge\psi(x_{k+\ell+1})\right)\right).
   \]
   Then $\lambda\in\LocFOplus[p{-}\ell,q]$.
   Since $\CA\models\lambda[\vec a^\ell_I]$, we have 
   $\CB\models\lambda[\vec b^\ell_I]$ by the induction hypothesis and
   equation \eqref{eq:localType}.
   Hence, there is a
   $b_{k+\ell+1}\in B$ such that $\dist^\CB(b_i,b_{k+\ell+1})\le d$ and
   $\dist^\CB(b_j,b_{k+\ell+1})> r_{\ell+1}$ for all 
   $j\in I$
   and
   $\CB\models\psi[b_{k+\ell+1}]$. This element $b_{k+\ell+1}$ is
   Duplicator's answer.

   We first show that invariant (II) holds for the new position
   $\vec a^{\ell+1},\vec b^{\ell+1}$.
   Recall that 
   $\dist^\CA(a_j,a_{k+\ell+1})>r_{\ell+1}$ for all $j\in[k{+}\ell]$ (since we are in Case~1b) 
   and $\dist^\CB(b_j,b_{k+\ell+1})>r_{\ell+1}$ for all $j\in I$ (due
   to our choice of $b_{k+\ell+1}$ according to $\lambda(\vec
   x_I)$). Consider an arbitrary
   $j\in[k{+}\ell]\setminus I$ and assume
   for contradiction that $\dist^{\CB}(b_j,b_{k+\ell+1})\leq
   r_{\ell+1}$. Then, 
   \[
     \dist^{\CB}(b_i,b_j)\leq 
     \dist^{\CB}(b_i,b_{k+\ell+1})+\dist^{\CB}(b_j,b_{k+\ell+1}) \leq
     d+r_{\ell+1}\leq
     r_{\ell}-r_{\ell+1}+r_{\ell+1} = r_\ell.
   \]
   By the induction hypothesis we obtain that $\dist^{\CA}(a_i,a_j)\leq r_\ell$, i.e., $G^{\CA}_{\vec a^\ell,r_\ell}$ contains the edge $\{i,j\}$, and hence $j\in I$ --- a contradiction!
   This (in combination with the induction hypothesis and the fact that $r_\ell\geq r_{\ell+1}$) proves that invariant (II) is satisfied.

   The argument that invariant (I) holds for the
   new position $\vec a^{\ell+1},\vec b^{\ell+1}$ follows
   the corresponding argument in Case~1a.
   Indeed, let $J'$
   be an $r_{\ell+1}$-component of $\vec a^{\ell+1}$.
   
   If $J'=\{k{+}\ell{+}1\}$ then $\vec a^{\ell+1}_{J'}=(a_{k+\ell+1})$ and 
   $\vec b^{\ell+1}_{J'}=(b_{k+\ell+1})$. Since
   $\CB\models \psi[b_{k+\ell+1}]$
   and
   $\psi=\bigwedge\ltp_{p-\ell-1,q}(\CA,a_{k+\ell+1})$, we have    
   $\ltp_{p-\ell-1,q}(\CA,a_{k+\ell+1})=\ltp_{p-\ell-1,q}(\CB,b_{k+\ell+1})$.
   
   If $J'\neq \{k{+}\ell{+}1\}$, then $J'\subseteq [k{+}\ell]$, and thus
   there is an $r_\ell$-component $J$ of $\vec a^\ell$ such that
   $J'\subseteq J$.
   Then the induction hypothesis $\ltp_{p-\ell,q}(\CA,\vec
   a^\ell_J)=\ltp_{p-\ell,q}(\CB,\vec b^\ell_J)$ implies
   $\ltp_{p-\ell-1,q}(\CA,\vec
   a^{\ell+1}_{J'})=\ltp_{p-\ell-1,q}(\CB,\vec b^{\ell+1}_{J'})$.

   \medskip\noindent
   \textit{Case 2:} In round $\ell{+}1$ of the game,
   Spoiler decides to make a far move.\\
   Without loss of generality,
   we assume that Spoiler picks
   $a_{k+\ell+1}\in A$ such that
   \[\dist^\CA(a_j,a_{k+\ell+1}) \ > \ r_\ell - r_{\ell+1}\] for all
   $j\in[k{+}\ell]$.

   Consider
   $\psi(x)\deff\bigwedge\ltp_{p-\ell-1,q}(\CA,a_{k+\ell+1})$
   and let $C$ be the set of all $a\in A$ such that
   $\CA\models\psi[a]$.
   Let $D$ be the set of all $b\in B$ such that $\CB\models\psi[b]$.

   If there is a $b_{k+\ell+1}\in D$ such that
   $\dist^{\CB}(b_j,b_{k+\ell+1})>r_{\ell+1}$ for all $j\in[k{+}\ell]$, this element $b_{k+\ell+1}$
   is Duplicator's answer. The argument that invariants  (I) and (II) hold for the
   new position $\vec a^{\ell+1},\vec b^{\ell+1}$ follows
   the corresponding arguments in Case~1b.

   \emph{Suppose for contradiction} that for all $d\in D$ there is a $j(d)\in
   [k{+}\ell]$ such that 
   $\dist^{\CB}(b_{j(d)},d)\le r_{\ell+1}$. Observe that
   for all $d,d'\in D$, if $j(d)=j(d')$ then
   \begin{equation}
     \label{eq:4}
     \dist^{\CB}(d,d')\ \le \
     \dist^{\CB}(d,b_{j(d)})+\dist^{\CB}(d',b_{j(d)})\ \le \
     2r_{\ell+1}. 
   \end{equation}
   
   \begin{claim}\label{claim:lem:main}
     There are $s,t,m\in\Nat$ that satisfy each of the following conditions:
     \begin{enumerate}
     \item[(a)]
     $r_{\ell+1}\le s= t-2r_{\ell+1}$
     \ and \ $2t\le r_\ell$ \ and \ $m\le k{+}\ell$;
     \item[(b)] there are $c_1,\ldots,c_m\in C$ such that
       $\dist^{\CA}(c_i,c_j)>2t$ for all 
       distinct 
       $i,j\in[m]$;
     \item[(c)] for all $c_1',\ldots,c_{m+1}'\in C$ there are distinct
       $i,j\in[m{+}1]$ such that $\dist^{\CA}(c'_i,c'_j)\le 2s$.
     \end{enumerate}

     \proof
     For $0\le i\le k{+}\ell$, we let $s^{(i)}\coloneqq
     r_{\ell+1}\cdot (1+2i)$.
     Moreover, we let
     $m^{(i)}$ be maximum such
     that there are $c_1,\ldots,c_{m^{(i)}}\in C$ with
     $\dist^{\CA}(c_\nu,c_{\nu'})>2s^{(i)}$ for all
     distinct
     $\nu,\nu'\in[m^{(i)}]$.
      If $C$ contains an infinite subset of elements of pairwise
      distance $>2s^{(i)}$, then we set $m^{(i)}\coloneqq\infty$. 
      Since $s^{(i)}\leq s^{(i+1)}$, we have $m^{(i)}\geq m^{(i+1)}$ for 
     $0\le i<k{+}\ell$.
     For $i=0$ and $n\deff \min\set{m^{(0)},k{+}p}$
     consider the sentence
     \[
       \xi \ \ \deff \ \
       \exists x_1 \cdots \exists x_{n} \Big(
         \bigwedge_{1\leq \nu<\nu'\leq n}
         \dist(x_\nu,x_{\nu'})>2s^{(0)} \ \wedge \
         \bigwedge_{1\leq \nu\leq n} \psi(x_\nu)
       \Big).
     \]
     Then, $\xi$ is a
     $(p{-}1,q,\rHalbe,k{+}p)$-independence sentence, because
     \[
     \myrho(p{-}\ell{-}1,q)\ = \ r_{\ell+1} \ = \ s^{(0)} \ \leq \ r_1
     \ \leq \ \frac{r_0}{2} 
     \ = \ \rHalbe,
     \]
     $n\leq k{+}p $,  and
     $\psi\in\LocFOplus[p{-}\ell{-}1,q]$.
     
     Since $n\le m^{(0)}$ we know that $\CA\models\xi$.
     By assumption \ref{item:ii:lem:main}, also
     $\CB\models\xi$.
     Hence, there
     are $d_1,\ldots,d_{n}\in D$ with
     $\dist^{\CB}(d_\nu,d_{\nu'})>2s^{(0)}$ for all distinct
     $\nu,\nu'\in[n]$. Then
     $j(d_\nu)\neq j(d_{\nu'})$, because by \eqref{eq:4}, $j(d_\nu)=j(d_{\nu'})$ implies
     $\dist^{\CB}(d_\nu,d_{\nu'})\le 2r_{\ell+1}= 2s^{(0)}$.
     Hence
     $n\leq k{+}\ell < k{+}p$, which implies that $n=m^{(0)}$.
     Thus
     \[
       k{+}\ell\ \ge\ m^{(0)}\ \ge\
       m^{(1)}\ \ge \ \cdots \ \ge\  m^{(k+\ell)}\ \ge\ 1.
     \]
     It follows that there is an $i\in[k{+}\ell]$ such that
     $m^{(i)}=m^{(i-1)}$. We let $m\coloneqq m^{(i)}$, $t\coloneqq s^{(i)}$, and
     $s\coloneqq s^{(i-1)}$. Then $m\le k{+}\ell$.
     From the definition of the $s^{(i)}$, we immediately obtain that
     $r_{\ell+1}\le s = t-2r_{\ell+1}$.
     Furthermore, recall that $r_\ell=\myrho(p{-}\ell,q)$ and $r_{\ell+1}=\myrho(p{-}\ell{-}1,q)$. Thus, according to condition \ref{myrho:constraint:inc} from Section~\ref{sec:FOpluspq} we have
     $r_\ell \geq r_{\ell+1}\cdot \big(2+4(q{-}(p{-}\ell))\big)$. By
     assumption we have $q\geq k{+}p$, and hence $q{-}(p{-}\ell)\geq k{+}\ell$. Therefore, 
     \[
2t \ \ = \ \  2s^{(i)} \ \ = \ \ r_{\ell+1}\cdot (2+4i)\ \ \leq \ \ r_{\ell+1}\cdot \big(2+4(k{+}\ell)\big)\ \ \leq \ \ r_{\ell}.    
     \]
     This proves assertion (a). Assertion (b) is immediate, and (c)
     follows from $s=s^{(i-1)}$, $m^{(i-1)}=m^{(i)}=m$ and the
     maximality of $m^{(i-1)}$.
     
     This completes 
     the proof of Claim~\ref{claim:lem:main}.
     \uend
   \end{claim}

 We choose $s,t,m$ according to Claim~\ref{claim:lem:main} and consider the sentence
 \[
   \xi \ \ \deff \ \
    \exists x_1 \cdots \exists x_{m} \Big(
         \bigwedge_{1\leq i<j\leq m}
         \dist(x_i,x_j)>2t \ \wedge \
         \bigwedge_{1\leq i\leq m} \psi(x_i)
       \Big).
     \]
     Then, $\xi$ is a
     $(p{-}1,q,\rHalbe,k{+}p)$-independence sentence, because
     \[
     \myrho(p{-}\ell{-}1,q) \ = \ r_{\ell+1} \ \leq \ t \
     \leq \ \frac{r_{\ell}}{2} \ \leq \ \frac{r_0}{2} \ = \ \rHalbe,
     \]
     $m\leq k{+}\ell \le k{+}p $,  and
     $\psi\in\LocFOplus[p{-}\ell{-}1,q]$.
     From Claim~\ref{claim:lem:main} we know that $\CA\models\xi$.
     By assumption \ref{item:ii:lem:main}, also $\CB\models\xi$.
    Hence, there
   are $d_1,\ldots,d_m\in D$ such that $\dist^{\CB}(d_i,d_{i'})>2t$ for
   all distinct $i,i'\in[m]$.
   Since $2t\ge 2 r_{\ell+1}$, by \eqref{eq:4} we have
   $j(d_i)\neq j(d_{i'})$ for all distinct $i,i'\in [m]$. Without loss of
   generality, we assume that $j(d_i)=i$ for all $i\in[m]$.

   Then for $i\in[m]$ we have $\dist^{\CB}(b_i,d_i)\le r_{\ell+1}$. Hence
   $\dist^\CB(b_i,b_{i'})>2t-2r_{\ell+1}$ for all distinct $i,i'\in[m]$. Since
   $2t-2r_{\ell+1}\le r_\ell$,
   it
   follows from the induction hypothesis (II) that
   \begin{equation}\label{eq:Distances}
     \dist^\CA(a_i,a_{i'}) \ \ > \ \ 2t-2r_{\ell+1}
   \end{equation}  
   for all distinct $i,i'\in[m]$.

   Consider the formula $\lambda(x)\coloneqq\exists y\big(\dist(x,y)\le
   r_{\ell+1}\wedge
   \psi(y)\big)$ and 
   note that $\lambda\in\LocFOplus[p-\ell,q]$
   (because $r_\ell\geq 2r_{\ell+1}$, and hence $r_{\ell+1}\leq
   r_\ell-r_{\ell+1} = \myrho(p{-}\ell,q)-\myrho(p{-}\ell{-}1,q)$).
   For all
   $i\in[m]$ we have $\CB\models \lambda[b_i]$; here the existential
   quantifier is witnessed by $d_i$. By 
   the induction hypothesis (I) and equation \eqref{eq:localType},
   we obtain that
   $\CA\models \lambda[a_i]$ for all $i\in[m]$.
   Thus there exist $c_1,\ldots,c_m\in C$ such that
   $\dist^\CA(a_i,c_i)\le r_{\ell+1}$ for all $i\in[m]$. Then for distinct $i,i'\in[m]$ we
   have
   \[
   \dist^{\CA}(a_i,a_{i'}) \ \leq \
   \dist^{\CA}(a_i,c_i)+\dist^{\CA}(c_i,c_{i'})+\dist^{\CA}(c_{i'},a_{i'})
   \ \leq \ \dist^{\CA}(c_i,c_{i'})+ 2r_{\ell+1},
   \]
   and hence
   \[
     \dist^{\CA}(c_i,c_{i'})\ \ge\ \dist^\CA(a_i,a_{i'})-2r_{\ell+1}
     \ \stackrel{\eqref{eq:Distances}}{>}\ 2t-4r_{\ell+1}\ = \ 2s.
   \]
   Furthermore, since we are in Case~2, we have 
   $\dist^{\CA}(a_{k+\ell+1},a_i)> r_\ell-r_{\ell+1}$ for all $i\in[m]$,
   and hence
   \[
     \dist^{\CA}(a_{k+\ell+1},c_i)\ > \ r_\ell-2r_{\ell+1}\ \ge \
     2t-2r_{\ell+1} \ \geq \ 2s.
   \]
   Thus, $c_1,\ldots,c_m,a_{k+\ell+1}$ are elements of $C$ of mutual
   distance $>2s$. This \emph{contradicts} Claim~\ref{claim:lem:main}(c).
   In summary, this completes the proof of Lemma~\ref{lem:main}.
\end{proof}

The following is an easy consequence of Lemma~\ref{lem:main}. It is interesting because $\CN_r^{\CA}(\vec a_J)\cong \CN_r^{\CB}(\vec b_J)$ for all connected components $J$ of $G^\CA_{\vec a,r}=G^{\CB}_{\vec b,r}$ does not imply $\CN_r^{\CA}(\vec a)\cong \CN_r^{\CB}(\vec b)$.

\begin{lemma}\label{lemma:ltpVsConnectedComponents}
Let $k,p,q\in\Nat$ with $k{+}p\leq q$ and $k\geq 1$. Let $r\deff
\myrho(p,q)$. Let $\CA,\CB$ be $\sigma$-structures such that
$\itp_{p-1,q,\rHalbe,k+p}(\CA)=\itp_{p-1,q,\rHalbe,k+p}(\CB)$.
For all $\vec
a=(a_1,\ldots,a_k)\in A^k$ and $\vec b=(b_1,\ldots,b_k)\in B^k$, the
following are equivalent.
\begin{eroman}
\item $G^{\CA}_{\vec a,r}=G^{\CB}_{\vec b,r}$ and for every
    $r$-component $J$ of $\vec a$ we have
    $\ltp_{p,q}(\CA,\vec a_J) = \ltp_{p,q}(\CB,\vec b_J)$.
\item $\ltp_{p,q}(\CA,\vec a) = \ltp_{p,q}(\CB,\vec b)$.
\end{eroman}
\end{lemma}
\begin{proof}
\emph{(i)$\Rightarrow$(ii)}: \ From Lemma~\ref{lem:main} we obtain
that $\tp_{p,q}(\CA,\vec a)=\tp_{p,q}(\CB,\vec b)$.
Using equations \eqref{eq:tpVsFormulas} and \eqref{eq:localType} and
$\LocFOplus[p,q]\subseteq\FOplus[p,q]$ yields
$\ltp_{p,q}(\CA,\vec a)=\ltp_{p,q}(\CB,\vec b)$.\medskip

\emph{(ii)$\Rightarrow$(i):} \
Since $\dist(x_i,x_j)\leq r$ and $\dist(x_i,x_j)>r$ can be
expressed by formulas in $\LocFOplus[p,q]$, from \emph{(ii)}
we obtain that $G^{\CA}_{\vec a,r}=G^{\CB}_{\vec b,r}$.
Furthermore, from \emph{(ii)} it follows that 
$\ltp_{p,q}(\CA,\vec a_J) = \ltp_{p,q}(\CB,\vec b_J)$ holds for every
$r$-component $J$ of $\vec a$.
\end{proof}

Combining the previous two lemmas with standard methods, we obtain the
following version of a \emph{Gaifman normal form theorem for $\FOplus[p,q]$}.

\begin{theorem}\label{thm:rpGaifman}
  Let $p,q\in\Nat$ with $p\leq q$. Let $r\deff \myrho(p,q)$.
  Every formula $\phi\in\FOplus[p,q]$ is equivalent to a
  formula $\phi'$ 
  that is a Boolean combination 
  of formulas in $\LocFOplus[p,q]$
  and of
  $(p{-}1,q,\rHalbe,k{+}p)$-independence sentences, where
  $k\coloneqq|\free(\phi)|$.
  Moreover, 
  $\free(\phi')=\free(\phi)$,
  and there is an algorithm that, for any $p,q\in\NN$ with $p\leq q$, upon input of $p,q,\phi$
  computes such a $\phi'$. 
\end{theorem}  
\begin{proof}
Without loss of generality we assume that
$\free(\phi)=\set{x_1,\ldots,x_k}$.
\medskip

We first consider the case $k\ge 1$.
By $\textup{Mod}(\phi)$ we denote the class of all pairs $(\CA,\vec
a)$ where $\CA$ is a $\sigma$-structure and $\vec a\in A^k$ such that
$\CA\models\phi[\vec a]$.

We let $T$ be the set of all pairs $(\xi,\lambda)$ such that
$\xi=\bigwedge\itp_{p-1,q,\rHalbe,k+p}(\CA)$ and
$\lambda=\bigwedge\ltp_{p,q}(\CA,\vec a)$ for some $(\CA,\vec a)$ in
$\textup{Mod}(\phi)$. Note that $T$ is finite, because there are only
finitely many $(p{-}1,q,\rHalbe,k{+}p)$-independence types and only finitely
many local $(k,p,q)$-types.

Let $S$ be the set of all $\xi$ that occur as the first component of
some pair in $T$ (i.e., there exists a $\lambda$ such that
$(\xi,\lambda)\in T$). For each $\xi\in S$ let $L_\xi$ be the set of
all $\lambda$ such that $(\xi,\lambda)\in T$. That is,
$T=\bigcup_{\xi\in S}(\set{\xi}\times L_\xi)$.
Let
\[
  \phi' \ \ \deff \ \
  \bigvee_{\xi\in S} \Big(\;
    \xi \, \wedge \,
    \bigvee_{\lambda\in L_\xi} \lambda
  \Big).
\]
Clearly, $\phi'$ is a formula of the desired shape. Next, we prove that
$\phi'$ is equivalent to $\phi$.
Consider an arbitrary $\sigma$-structure $\CB$ and a tuple $\vec b\in
B^k$. Our aim is to show that $\CB\models\phi[\vec b] \iff
\CB\models\phi'[\vec b]$.

If $\CB\models\phi[\vec b]$, then $(\CB,\vec b)$ belongs to
$\textup{Mod}(\phi)$, and hence $T$ contains the tuple $(\xi,\lambda)$
for
$\xi\deff \bigwedge\itp_{p-1,q,\rHalbe,k+p}(\CB)$ and
$\lambda\deff\bigwedge\ltp_{p,q}(\CB,\vec b)$.
Thus, $\xi\in S$ and $\lambda\in L_\xi$. Since $\CB\models\xi$ and
$\CB\models\lambda[\vec b]$ we obtain that $\CB\models\phi'[\vec b]$.

For the opposite direction assume that $\CB\models\phi'[\vec b]$. Then
there exists a $\xi\in S$ and a $\lambda\in L_\xi$ such that
$\CB\models\xi$ and $\CB\models\lambda[\vec b]$.
Since $(\xi,\lambda)\in T$, there exists a $\sigma$-structure $\CA$
and a tuple $\vec a\in A^k$ such that $(\CA,\vec a)$ belongs to
$\textup{Mod}(\phi)$ and
$\xi=\bigwedge\itp_{p-1,q,\rHalbe,k+p}(\CA)$ and
$\lambda=\bigwedge\ltp_{p,q}(\CA,\vec a)$.
Since $\CB\models\xi$ and $\CB\models\lambda[\vec b]$, we obtain from
\eqref{eq:itp} and \eqref{eq:localType} that
$\itp_{p-1,q,\rHalbe,k+p}(\CA)=\itp_{p-1,q,\rHalbe,k+p}(\CB)$ and
$\ltp_{p,q}(\CA,\vec a)=\ltp_{p,q}(\CB,\vec b)$. By applying
Lemma~\ref{lemma:ltpVsConnectedComponents}, we obtain that the
assumptions of Lemma~\ref{lem:main} are satisfied. Hence,
Lemma~\ref{lem:main} yields $\tp_{p,q}(\CA,\vec a)=\tp_{p,q}(\CB,\vec
b)$.
Using equation \eqref{eq:tpVsFormulas} and the fact that $\CA\models\phi[\vec a]$,
we obtain that $\CB\models\phi[\vec b]$.

This proves that an equivalent formula of the desired shape indeed exists.
Using standard techniques, one obtains an algorithm that computes upon
input of $\phi$ an equivalent formula $\phi'$ of the desired shape. For this, we regard $\FO^+$ formulas as $\FO$ formulas, replacing the distance atoms by distance formulas. Then we can use the completeness theorem to recursively
enumerate all valid formulas (recall that a
formula $\chi(\vec x)$ is valid if and only if for every $\sigma$-structure
$\CA$ and every tuple $\vec a\in A^k$ we have $\CA\models\chi[\vec a]$). For
each enumerated formula $\chi$ we check if it is of the form $(\phi
\leftrightarrow \phi')$ where $\phi'$ is a formula of the form
$\bigvee_{\xi\in S'}\big(\xi \wedge \bigvee_{\lambda\in L'_\xi} \lambda \big)$
for a set $S'$ of formulas $\xi=\bigwedge\Xi$, where $\Xi$ is a
$(p{-}1,q,\rHalbe,k{+}p)$-independence type,  and sets
$L'_\xi$ of formulas $\lambda\coloneqq\bigwedge\Lambda$ for a local $(k,p,q)$-type $\Lambda$.
If so, we terminate and output $\phi'$.
Since we know that there exists a formula of this shape which is
equivalent to $\phi$, the algorithm will terminate on every input
formula $\phi$. This proves the theorem for the case $k\geq 1$.
\medskip

The case $k=0$ can be proved analogously. Let
$\textup{Mod}(\phi)$ be the class of all $\sigma$-structures $\CA$
such that $\CA\models\phi$. Let $S$ be the set of all $\xi$ such
that
$\xi=\bigwedge\itp_{p-1,q,\rHalbe,k+p}(\CA)$ for some $\CA$ in
$\textup{Mod}(\phi)$. Choose $\phi'\deff \bigvee_{\xi\in
  S}\xi$. Analogously as in the previous case, we can prove that
$\phi'$ is equivalent to $\phi$. This proves that an equivalent formula
of the desired shape indeed exists. Using the same techniques as in
the previous case, we also obtain an algorithm that upon input of
$\phi$ outputs such a $\phi'$.

In summary, this completes the proof of Theorem~\ref{thm:rpGaifman}.
\end{proof}

\section{A Refined Version of Theorem~\ref{thm:rpGaifman}}
\label{sec:FixingThm71}

Using a slight modification of the proof of Theorem~\ref{thm:rpGaifman}, we can show the following
Theorem~\ref{thm:forGroheS18}, which is similar to \cite[Theorem~7.1]{GroheS18}, yet, its
statement is conceptually much simpler.

Theorem~\ref{thm:forGroheS18} goes beyond
Theorem~\ref{thm:rpGaifman} in two ways, both of which are important
for applications such as counting and enumeration (cf.\
\cite{GroheS18,SchweikardtSV22}): first, the 
normal form is a disjunction of mutually exclusive formulas (assertion
\ref{item:forGroheS18:two} of the theorem), and second, the free
variables in the local formulas are always interpreted by elements
that are close to each other in the Gaifman graph of the structure (they form an $r$-connected component). 

For $k\in\NNpos$ we write $\CG_k$ to denote the set of all undirected
graphs with vertex set $[k]$.
For $p,q,\hat{r},s\in\NN$ with $p\leq q$, a \emph{basic local $(p,q,\hat{r},s)$-sentence} is an $\FOplus$ sentence of the form
\eqref{eq:basic-local-Gaifman}, 
where $m\leq s$, $r\leq \hat{r}$, and $\psibl(x)\in\FOplus[p,q]$ is $r$-local.

\begin{theorem}\label{thm:forGroheS18}
There is an algorithm which, upon input of $p,q\in\NN$ with $p< q$ and
a formula $\phi(x_1,\ldots,x_k)\in\FOplus[p,q]$ with $k= |\free(\phi)|\geq 1$,
lets $r\deff\myrho(p,q)$ and computes
for each $G\in \CG_k$ a number $m_G\in\NN$ and for each $i\in[m_G]$
\begin{itemize}
\item a Boolean combination $\xi^i_G$ of 
basic local $(p{-}1,q,\rHalbe,k{+}p)$-sentences,
  and
\item for each connected component $I$ of $G$ an $r$-local formula
  $\psi^i_{G,I}(\vec x_I)\in\FOplus[p,q]$
\end{itemize}
such that the following holds.
\begin{enumerate}
\item\label{item:forGroheS18:one}
  For all $\sigma$-structures $\CA$ and all $\vec a\in A^k$ we have
  $\CA\models\phi[\vec a]$ if and only if for $G\deff G_{\vec
    a,r}^{\CA}$ there is an $i\in[m_G]$ such that
  $\CA\models\xi_G^i$ and $\CA\models
  \psi_{G,I}^i[\vec{a}_I]$ for every connected component $I$
  of $G$.
\item\label{item:forGroheS18:two} 
  For all $\sigma$-structures $\CA$ and all $\vec a\in A^k$, there is
  at most one $i\in [m_G]$ for $G\deff G^{\CA}_{\vec a,r}$ such that  
  $\CA\models\xi_G^i$ and $\CA\models
  \psi_{G,I}^i[\vec{a}_I]$ for every connected component $I$
  of $G$.
\end{enumerate}  
\end{theorem}
\begin{proof}
By $\textup{Mod}(\phi)$ we denote the class of all pairs $(\CA,\vec
a)$ where $\CA$ is a $\sigma$-structure and $\vec a\in A^k$ such that
$\CA\models\phi[\vec a]$.
For every $G\in\CG_k$ we let $\conn(G)$ be the set consisting of the
connected components of $G$.
We let
\[
  T \ \deff \
  \{\, (G,\Xi,(\Lambda_I)_{I\in \conn(G)})\  :
  \begin{array}[t]{l}
    (\CA,\vec
    a)\in\textup{Mod}(\phi),\ \
    G=G^{\CA}_{\vec a,r}, \ \
    \Xi=\itp_{p-1,q,\rHalbe,k+p}(\CA),
\\
     \text{ and \ } \Lambda_I=\ltp_{p,q}(\CA,\vec a_I), \text{ for each } I\in \conn(G) \, \}.
   \end{array}
\]    
Note that $T$ is finite, because there are only
finitely many graphs with vertex set $[k]$,
only finitely many $(p{-}1,q,\rHalbe,k{+}p)$-independence types, 
and only finitely many local $(k',p,q)$-types with $1\leq k'\leq k$.

Let $S$ be the set of all $G\in\CG_k$ that occur as the first component of
some tuple in $T$. For each $G\in S$ let $L_G$ be the set of
all $(\Xi,(\Lambda_I)_{I\in \conn(G)})$ such that
$(G,\Xi,(\Lambda_I)_{I\in \conn(G)})\in T$. That is,
$T=\bigcup_{G\in S}(\set{G}\times L_G)$.

For $G\in S$ let
\begin{equation}\label{eq:deltaGr}
   \delta_{G,r}(x_1,\ldots,x_k) \ \deff \ \bigwedge_{\set{i,j}\in
     E(G)}\dist(x_i,x_j)\leq r \ \wedge \ \bigwedge_{i<j,\atop \set{i,j}\not\in
     E(G)}\dist(x_i,x_j)>r.
 \end{equation}
 Then $\delta_{G,r}\in\LocFOplus[p,q]$, and for every
 $\sigma$-structure
$\CA$ and every tuple $\vec a\in A^k$ we
have
$\CA\models\delta_{G,r}[\vec a] \iff G^{\CA}_{\vec a,r}=G$.
Let
\[
  \phi' \ \ \deff \quad
  \bigvee_{G\in S} \Big(\;
    \delta_{G,r} \ \wedge 
    \bigvee_{(\Xi,(\Lambda_I)_{I\in \conn(G)})\in L_G} \big(
    \bigwedge\Xi \,\wedge\,\bigwedge_{I\in \conn(G)} \bigwedge\Lambda_I
    \big) 
  \Big).
\]

Next, we prove that $\phi'$ is equivalent to $\phi$.
Consider an arbitrary $\sigma$-structure $\CB$ and a tuple $\vec b\in
B^k$. Our aim is to show that $\CB\models\phi[\vec b] \iff
\CB\models\phi'[\vec b]$.

If $\CB\models\phi[\vec b]$, then $(\CB,\vec b)$ belongs to
$\textup{Mod}(\phi)$, and hence $T$ contains the tuple
\allowbreak
$(G,\Xi,(\Lambda_I)_{I\in \conn(G)})$
for $G\deff G^{\CB}_{\vec b,r}$,
$\Xi\deff \itp_{p-1,q,\rHalbe,k+p}(\CB)$ and
$\Lambda_I\deff\ltp_{p,q}(\CB,\vec b_I)$, for every $I\in \conn(G)$.
Thus, $G\in S$ and $(\Xi,(\Lambda_I)_{I\in \conn(G)})\in L_G$.
Since $\CB\models\delta_{G,r}[\vec b]$ and $\CB\models\bigwedge\Xi$ and
$\CB\models\bigwedge\Lambda_I[\vec b_I]$ for every $I\in \conn(G)$, we obtain that $\CB\models\phi'[\vec b]$.

For the opposite direction assume that $\CB\models\phi'[\vec b]$. Then
there exists a $G\in S$ and a tuple $(\Xi,(\Lambda_I)_{I\in \conn(G)})\in L_G$ such that
$\CB\models\delta_{G,r}[\vec b]$ (i.e., $G=G^{\CB}_{\vec b,r}$) and $\CB\models \bigwedge\Xi$ and
$\CB\models\bigwedge\Lambda_I[\vec b_I]$ for every $I\in \conn(G)$.

Since $(G,\Xi,(\Lambda_I)_{I\in \conn(G)})\in T$, there exists a $\sigma$-structure $\CA$
and a tuple $\vec a\in A^k$ such that $(\CA,\vec a)$ belongs to
$\textup{Mod}(\phi)$ and
$G=G^{\CA}_{\vec a,r}$ and
$\Xi=\itp_{p-1,q,\rHalbe,k+p}(\CA)$ and
$\Lambda_I=\ltp_{p,q}(\CA,\vec a_I)$, for every $I\in \conn(G)$.
Since $\CB\models\bigwedge\Xi$ and $\CB\models\bigwedge\Lambda_I[\vec
b_I]$, for every $I\in \conn(G)$, we obtain from
\eqref{eq:itp} and \eqref{eq:localType} that
$\itp_{p-1,q,\rHalbe,k+p}(\CA)=\itp_{p-1,q,\rHalbe,k+p}(\CB)$ and
$\ltp_{p,q}(\CA,\vec a_I)=\ltp_{p,q}(\CB,\vec b_I)$ for every $I\in
\conn(G)$. Since also $G^{\CB}_{\vec b,r}=G=G^{\CA}_{\vec a,r}$, 
we obtain that the
assumptions of Lemma~\ref{lem:main} are satisfied. Hence,
Lemma~\ref{lem:main} yields $\tp_{p,q}(\CA,\vec a)=\tp_{p,q}(\CB,\vec
b)$.
Using equation \eqref{eq:tpVsFormulas} and the fact that $\CA\models\phi[\vec a]$,
we obtain that $\CB\models\phi[\vec b]$.

This proves that $\phi'$ is equivalent to $\phi$.
Using standard techniques, we can utilise this to obtain an algorithm that computes upon
input of $\phi$ an equivalent formula $\phi'$ of the form
\begin{equation}\label{eq:bugfix}
  \bigvee_{G\in S'} \Big(\;
    \delta_{G,r} \ \wedge 
    \bigvee_{(\Xi,(\Lambda_I)_{I\in \conn(G)})\in L'_G} \big(
    \bigwedge\Xi \,\wedge\,\bigwedge_{I\in \conn(G)} \bigwedge\Lambda_I
    \big) 
  \Big),
\end{equation}
for a set $S'\subseteq \CG_k$ and, for every $G\in S'$, a set $L'_G$
of tuples of the form $(\Xi,(\Lambda_I)_{I\in \conn(G)})$, where $\Xi$ is
a $(p{-}1,q,\rHalbe,k{+}p)$-independence type, and
$\Lambda_I$ is a local $(|I|,p,q)$-type, for each $I\in \conn(G)$.

For this, we regard $\FO^+$ formulas as $\FO$ formulas, replacing the
distance atoms by distance formulas. Then we can use the completeness
theorem to recursively enumerate all valid formulas (recall that a
formula $\chi(\vec x)$ is valid if and only if for every $\sigma$-structure
$\CA$ and every tuple $\vec a\in A^k$ we have $\CA\models\chi[\vec a]$). For
each enumerated formula $\chi$ we check if it is of the form $(\phi
\leftrightarrow \phi')$ where $\phi'$ is a formula of the form
\eqref{eq:bugfix}.
If so, we terminate and output $\phi'$.
Since we know that there exists a formula of this shape which is
equivalent to $\phi$, the algorithm will terminate on every input
formula $\phi$. 

Finally, to finish the proof of Theorem~\ref{thm:forGroheS18}, after having computed such a $\phi'$, we let
$m_G\deff |L'_G|$ for all $G\in S'$, and
$m_G\deff 0$ for all $G\in\CG_k\setminus S'$.
For all $G\in S'$ we let $t_1,\ldots,t_{m_G}$ be a list of all tuples in $L'_G$, and for each $i\in[m_G]$ we let $(\Xi_G^i, (\Lambda_{G,I}^i)_{I\in \conn(G)})=t_i$ and output
\[
  \xi_G^i \ \deff \ \bigwedge\Xi_G^i
\qquad\text{and}\qquad
  \psi_{G,I}^i(\vec x_I) \ \deff \ \bigwedge\Lambda_{G,I}^i\;(\vec x_I),\quad\text{for all } I\in \conn(G).
\]

Since $\phi'$ is equivalent to $\phi$, we obtain that assertion \ref{item:forGroheS18:one} of Theorem~\ref{thm:forGroheS18} is satisfied.
Concerning assertion \ref{item:forGroheS18:two}, \emph{assume for contradiction} that there is a $\sigma$-structure $\CA$, a tuple $\vec a\in A^k$, and two distinct $i,j\in[m_G]$, for $G=G^{\CA}_{\vec a,r}$, such that $(\CA,\vec a)$ satisfies
\[
  \xi^i_{G} \ \wedge  \bigwedge_{I\in \conn(G)} \psi^i_{G,I}(\vec x_I)
  \qquad\text{and}\qquad
  \xi^j_{G} \ \wedge  \bigwedge_{I\in \conn(G)} \psi^j_{G,I}(\vec x_I),
\]
i.e., $(\CA,\vec a)$ satisfies
\[
  \bigwedge \Xi_i \ \wedge  \bigwedge_{I\in \conn(G)}\bigwedge\Lambda_{G,I}^{i}(\vec x_I)
  \qquad\text{and}\qquad
  \bigwedge \Xi_j \ \wedge  \bigwedge_{I\in \conn(G)}\bigwedge\Lambda_{G,I}^{j}(\vec x_I).
\]
Since $[m_G]\neq\emptyset$, we know that $G\in S'$.
Since $\Xi_i$ and $\Xi_j$ are \emph{$(p{-}1,q,\rHalbe,k{+}p)$-independence types}, and any structure has exactly one such type, we obtain that $\Xi_i=\Xi_j$.
Furthermore, for every $I\in \conn(G)$, since $\Lambda^i_{G,I}$ and $\Lambda^j_{G,I}$ are
\emph{local $(|I|,p,q)$-types}, and any $|I|$-tuple of a structure has exactly one such type,
we obtain that $\Lambda^i_{G,I}=\Lambda^j_{G,I}$.
Thus, $t_i=t_j$, contradicting the assumption that $i\neq j$ (because for distinct $i,j\in[m_G]$,  $t_i$ and $t_j$ are tuples in $L'_G$ with $t_i\neq t_j$).

Finally, note that $\xi^i_G$ is a Boolean combination of
$(p{-}1,q,\rHalbe,k{+}p)$-independence sentences, and that every
$(p{-}1,q,\rHalbe,k{+}p)$-independence sentence is a
basic local $(p{-}1,q,\rHalbe,\allowbreak{}k{+}p)$-sentence.
Furthermore, for each $I\in \conn(G)$, the formula
$\psi^i_{G,I}(\vec x_I)$ is a
formula in $\LocFOplus[p,q]\subseteq\FOplus[p,q]$.
From Lemma~\ref{lemma:locality} we obtain that $\psi^i_{G,I}(\vec x_I)$ is $\myrho(p,q)$-local; and we know that
$\myrho(p,q)=r$.
This completes the proof of Theorem~\ref{thm:forGroheS18}.
\end{proof}

\section{A Rank-Preserving Gaifman Normal Form Theorem}
\label{sec:rank-new}

We define the \emph{rank} $\rk(\phi)$ of an $\FOplus$ formula $\phi$
as follows.

\begin{definition}\label{def:rank}
  The \emph{rank} $\rk(\phi)$ of an $\FOplus$ formula $\phi$ is the
  least $q\in\NN$ such that $\phi\in\FOplus[p,q]$ for some $p\leq q$.
\uend
\end{definition}

Using Remark~\ref{remark:DefOfFOpq}, one obtains that
$\rankletter\deff\rk(\phi)$ is the least $q\in\NN$ such that
$\phi\in\FOplus\big[q{-}|\free(\phi)|,q\big]$, and 
that $\phi\in\FOplus[\rankletter{-}|\free(\phi)|,q]$,  for
every $q\in\NN$ with $q\geq \rankletter$.

To formulate our ``rank-preserving normal form theorem'', we use the
notation that was introduced directly after Theorem~\ref{thm:gaifman},
and generalise it to $\FOplus$ as follows. 
A \emph{basic local sentence} in $\FOplus$ is an $\FOplus$ sentence of
the form \eqref{eq:basic-local-Gaifman},
where $\psibl(x)$ is an $r$-local $\FOplus$ formula with one free
variable.
An $\FOplus$ formula in Gaifman normal form is a Boolean combination
of local $\FOplus$ formulas and of basic local sentences in $\FOplus$.

The \emph{inner rank} of a basic local sentence $\xi$ of the form \eqref{eq:basic-local-Gaifman}
is the rank $\rk(\psibl)$.
Let $\phi'$ be an $\FOplus$ formula in Gaifman normal form, and let $S$ and $L$
be the sets of basic local sentences and of local formulas,
respectively, such that $\phi'$ is a Boolean combination of the
formulas in $S\cup L$.
The \emph{inner rank} of  $\phi'$
is the maximum of the inner ranks of all sentences in $S$.
The \emph{outer rank} of $\phi'$ is the maximum of the
ranks of all the formulas in~$L$.

As a corollary to Theorem~\ref{thm:rpGaifman}, we obtain the following
``rank-preserving Gaifman normal form theorem''.

\begin{theorem}\label{thm:rankGaifman}\label{cor:rpGaifman}
  Every
  $\FO^+$ formula $\phi$ of rank $\rankletter\coloneqq\rk(\phi)$ and with
  $k\deff|\free(\phi)|$ is equivalent to an
  $\FOplus$ formula $\phi'$ in Gaifman normal form
  of outer rank,  inner rank, and 
  width at most $\rankletter$, and of radius at most
  $\min\set{\myrho(\rankletter{-}k,\rankletter),\frac{\myrho(\rankletter,\rankletter)}{2}}$.
  Moreover,  $\free(\phi')=\free(\phi)$, and there is an algorithm that, upon
  input of $\phi$, computes such a $\phi'$.
\end{theorem}

\begin{proof}
Upon input of $\phi$, we compute $\rankletter\deff\rk(\phi)$ and
$k\deff|\free(\phi)|$, and we let $q\deff\rankletter$ and $p\deff q{-}k$.
Then, according to Definition~\ref{def:rank}, $\phi\in\FOplus[p,q]$.

We use the algorithm provided by Theorem~\ref{thm:rpGaifman} to
compute a formula $\phi'$ that is a Boolean combination of formulas in
$\LocFOplus[p,q]$ and of
$(p{-}1,q,\rHalbe,k{+}p)$-independence sentences,
for $r\deff\myrho(p,q)=\myrho(\rankletter{-}k,\rankletter)$.
By Theorem~\ref{thm:rpGaifman}, $\phi'$ is equivalent to $\phi$ and $\free(\phi')=\free(\phi)$.

According to Definition~\ref{def:rank}, each of the formulas in
$\LocFOplus[p,q]$ has rank $\leq q=\rankletter$, and it is $r$-local by
Lemma~\ref{lemma:locality}. Furthermore, if $k\geq 1$ then
$r=\myrho(\rankletter{-}k,\rankletter)\leq
\frac{\myrho(\rankletter,\rankletter)}{2}$ (this follows because
$\myrho$ satisfies condition \ref{myrho:constraint:inc} formulated at
the beginning of Section~\ref{sec:FOpluspq}).

Furthermore, by definition (cf., Section~\ref{subsec:IndependenceSentences}), each $(p{-}1,q,\rHalbe,k{+}p)$-independence
sentence is a basic local sentence of width $\leq k{+}p =\rankletter$,
of inner rank $\leq q=\rankletter$, and of radius $\leq \rHalbe = \frac{\myrho(\rankletter-k,\rankletter)}{2}\leq\frac{\myrho(\rankletter,\rankletter)}{2}$.

Thus, $\phi'$ is an $\FOplus$ formula in Gaifman normal form. Moreover, $\phi'$ has
inner rank $\leq \rankletter$, width $\leq \rankletter$, outer rank
$\leq \rankletter$, and radius $\leq \min\set{\myrho(\rankletter{-}k,\rankletter),\frac{\myrho(\rankletter,\rankletter)}{2}}$.
This completes the proof of Theorem~\ref{thm:rankGaifman}.
\end{proof}

\section{Application: Deciding First-Order Properties of Nowhere Dense Structures}
\label{sec:app}

As an application of our new locality theorem, we give a simplified proof of the following theorem.

\begin{theorem}[Grohe, Kreutzer, Siebertz~\cite{GroheKS17}]\label{theo:GKS}
  Let $\CC$ be a nowhere dense class of $\sigma$-structures. Then for every $\FO$ sentence $\phi$ and every $\epsilon>0$ there is an algorithm that, given $\CA\in\CC$, decides if $\CA\models\phi$ in time $O(|\CA|^{1+\epsilon})$.
\end{theorem}

Grohe et al.~\cite{GroheKS17} also prove a uniform version of the theorem for effectively nowhere dense classes $\CC$. To keep things simple, here we focus on the non-uniform version (i.e., Theorem~\ref{theo:GKS}). But our simplified proof can easily be adapted to the uniform version as well.

While we review the main definitions necessary to understand the simplification of the main algorithm that we achieve, this section is not self-contained; we rely on several substantial results from \cite{GroheKS17} that we state here without proof.

\emph{In this section, we only consider finite graphs and finite
 $\sigma$-structures for a finite relational vocabulary $\sigma$.}
Without loss of generality we assume that $\sigma$ contains a relation
symbol $R_0$
of arity 0. We let $\logic{true}\deff R_0()\vee\neg R_0()$ and
$\logic{false}\deff\neg\,\logic{true}$. Clearly, $\logic{true}$ is valid,
$\logic{false}$ is unsatisfiable, and both are sentences that belong
to $\LocFOplus[p,q]$ for all $p,q\in\NN$ with $p\leq q$.

Remember that $|\CA|\coloneqq|A|$ denotes the order of a structure
$\CA$. The \emph{size} of a $\sigma$-structure $\CA$ is
\[
\|\CA\| \ \ \coloneqq \ \ |A|+\sum_{R\in\sigma}|R^{\CA}|.
\]
Accordingly, for a graph $G$ we let $|G|\deff
|V(G)|$ and $\|G\|\deff |V(G)|+|E(G)|$.

Note that the Gaifman graph $G_\CA$ of $\CA$ can be computed from $\CA$ in time $O(\|\CA\|)$ (assuming the vocabulary $\sigma$ is fixed). This will allow us to always compute the Gaifman graphs of the structures we are dealing with within the permitted time constraints.

\subsection{The Splitter Game and Nowhere Dense Classes}

To define nowhere dense graph classes, we directly use a game characterisation of such classes introduced in \cite[Section~4]{GroheKS17}. 
Let $G$ be a graph and $\ell,r\in\Nat$. 
The~$(\ell,r)$-\emph{splitter game} on~$G$ is played by two players called \emph{Connector} and
\emph{Splitter} as follows.
We let~$G^{(0)}\coloneqq G$. For $i\ge1$, in round $i$ of the game, Connector selects a vertex~$v^{(i)}\in V(G^{(i-1)})$. Then Splitter
selects a set $W^{(i)}\subseteq N_r^{G^{(i-1)}}(v^{(i)})$ such that $\sum_{j=1}^{i}|W^{(j)}|\le\ell$. If $N_r^{G^{(i-1)}}(v^{(i)})\setminus
W^{(i)}=\emptyset$, then Splitter wins. Otherwise, the
play continues with 
\[
  G^{(i)} \quad \deff\quad G^{(i-1)}\big[N_r^{G^{(i-1)}}(v^{(i)})\setminus
W^{(i)}\big].
\]
If the play never ends, which is possible because Splitter may always choose $W^{(i)}=\emptyset$, then Connector wins.

It is convenient to allow Splitter to choose $W^{(i)}=\emptyset$, even though this does not help her make progress: if $W^{(i)}=\emptyset$ then Connector can always choose $v^{(i+1)}=v^{(i)}$, and the position remains unchanged. Note that the sets $W^{(i)}$ selected by Splitter are mutually disjoint. Hence in every play there are at most $\ell$ rounds $i$ where Splitter chooses a non-empty set $W^{(i)}$.

A class $\CC$ of graphs is \emph{nowhere dense} if for every $r\in\Nat$ there is an $\ell(r)\in\Nat$ such that for all graphs $G\in\CC$, Splitter has a winning strategy for the $(\ell(r),r)$-splitter game on~$G$. Furthermore, a class $\CC$ of $\sigma$-structures is \emph{nowhere dense} if the class $\{G_{\CA}\mid\CA\in\CC\}$ of the Gaifman graphs of all structures in $\CC$ is nowhere dense. 

It is a well known fact that if $\CC$ is a nowhere dense class of
graphs then for every $\epsilon>0$ and all $G\in\CC$ it holds that
$\|G\|=O(|G|^{1+\epsilon})$ (cf.\ \cite{GroheKS17}). Hence if $\CC$ is
a nowhere dense class of $\sigma$-structures then
$\|\CA\|=O(|A|^{1+\epsilon})$ for all $\CA\in\CC$. Furthermore, the
closure of a nowhere dense class of graphs under taking subgraphs is
nowhere dense as well, because if Splitter has a winning strategy for
the $(\ell,r)$-splitter game on a graph $G$, then this winning
strategy induces a winning strategy for the game on any subgraph of $G$.

We actually need to be able to compute winning strategies for Splitter efficiently. And not only that, we must be able to do this while moving to subgraphs of the original graph. 
We will describe a generic strategy for Splitter that can be computed efficiently using just shortest path computations.
A \emph{partial play} of the $(\ell,r)$-splitter game on a graph $G$ is a sequence 
\[
\pi \ \ \coloneqq \ \ (v^{(1)},W^{(1)},v^{(2)},W^{(2)},\ldots,v^{(k)},W^{(k)}),
\]
where $v^{(i)}$, $W^{(i)}$ are the choices made by Connector and Splitter in round $i$ of the play. The \emph{length} of $\pi$ is $k$.
We let $G^{(0)}\coloneqq G$ and $G^{(i)}\coloneqq G^{(i-1)}\big[N_r^{G^{(i-1)}}(v^{(i)})\setminus
W^{(i)}\big]$ for $i\in[k]$. 
Observe that by the rules of the game, for all $i\in[k]$ we have
$v^{(i)}\in V(G^{(i-1)})$, which implies $V(G^{(j)})\neq\emptyset$ for $j\in[0,k{-}1]$, and $W^{(i)}\subseteq  N_r^{G^{(i-1)}}(v^{(i)})$. 

The following lemma is extracted from the proof of Theorem~4.2 and
Remark~4.3 of \cite{GroheKS17}.
For the reader's convenience, we give a proof of the lemma.

\begin{lemma}[\cite{GroheKS17}]\label{lem:splitter-strategy}
  Let $\CC$ be a nowhere dense class of graphs. Then for every
  $r\in\Nat$ there is a $t=t(\CC,r)\in\PNat$ such that the following
  holds. Let $G\in\CC$, and let
  $\pi=(v^{(1)},W^{(1)},v^{(2)},W^{(2)},\ldots,v^{(k)},W^{(k)})$,
  where
  $k\in\NN$ and
  $v^{(i)}\in V(G)$ and $W^{(i)}\subseteq V(G)$ for all
  $i\in[k]$. Let  
  \begin{equation}\label{eq:splitter-strategy}
    G^{(0)} \ \coloneqq \ G \ \supseteq \ H^{(1)} \ \supseteq \
    G^{(1)} \ \supseteq \ \cdots \ \supseteq \ H^{(k)} \ \supseteq \ G^{(k)}
  \end{equation}
   such that the following conditions are satisfied.
  \begin{eroman}
  \item\label{item:splitter-strategy-i} For all $i\in[k]$ it holds that $v^{(i)}\in V(H^{(i)})$.
  \item\label{item:splitter-strategy-ii} $W^{(1)}=\emptyset$
    and for all $i\ge 2$ it holds that $W^{(i)}= N_r^{H^{(i)}}(v^{(i)})\cap\bigcup_{j=1}^{i-1}V(P^{(i)}_j)$,
    where $P^{(i)}_j$ is a path of length at most $r$ from $v^{(j)}$ to $v^{(i)}$ in the graph $H^{(j)}$.
  \item\label{item:splitter-strategy-iii} For all $i\ge 1$ it holds
    that
    $G^{(i)}= H^{(i)}\big[N_r^{H^{(i)}}(v^{(i)})\setminus W^{(i)}\big]$.
  \end{eroman}
  Then $k\le t$.
\end{lemma}
\begin{proof}
We use the following characterisation,  due to \cite{NesetrilO11}, of
nowhere dense graph classes in terms of \emph{uniform quasiwideness}
(cf.~\cite{Dawar10}): \emph{A class $\CC$ of graphs is nowhere dense
  if and only if there are functions $s\colon\Nat\to\Nat$ and
  $N\colon\Nat\times\Nat\to\Nat$ such that for all $r,m\in\Nat$, all
  $G\in\CC$, and all $X\subseteq V(G)$ of size $|X|>N(r,m)$
  there is a
  set  $S\subseteq V(G)$ of size $|S|\le s(r)$ and a set $Y\subseteq
  X$ of size $|Y|\ge m$ with $\dist^{G\setminus S}(y,y')>r$ for all
  distinct $y,y'\in Y$.} 

We choose the functions $s,N$ for our nowhere dense class $\CC$
according to this characterisation and let $t\coloneqq
N(r,2s(r){+}2)$. Suppose for contradiction that
$k> t$.

Let $X\coloneqq\{v^{(1)},\ldots,v^{(k)}\}$. Choose $S\subseteq V(G)$ of size $|S|\le s(r)$ and a set $Y\subseteq X$ of size $|Y|=2s(r){+}2$ with $\dist^{G\setminus S}(y,y')>r$ for all distinct $y,y'\in Y$. Say, 
\[
Y\ = \ \{\,v^{(i_1)},\ldots,v^{(i_{2s(r)+2})}\,\}\quad\text{with}\quad1\le i_1<i_2<\cdots<i_{2r(s)+2}\le k.
\]
For all $j\in[s(r){+}1]$, let $P_j\coloneqq P^{(i_{2j})}_{i_{2j-1}}$
be the path of length at most $r$ from $v^{(i_{2j-1})}$ to
$v^{(i_{2j})}$ in $H^{(i_{2j-1})}$ from condition
\ref{item:splitter-strategy-ii}. Then $V(P_j)\cap
V(G^{(i_{2j})})=\emptyset$ and therefore $V(P_j)\cap
V(H^{(i_{2j+1})})=\emptyset$ (if $j<s(r){+}1$). This means that the
paths $P_1,\ldots,P_{2s(r)+1}$ are mutually vertex disjoint. Since for
every $j$ the path $P_j$ has length at most $r$ and $\dist^{G\setminus
  S}(v^{(i_{2j-1})},v^{(i_{2j})})>r$, it holds that $V(P_j)\cap
S\neq\emptyset$. Thus $|S|\ge s(r){+}1$, which is a contradiction.
This completes the proof of Lemma~\ref{lem:splitter-strategy}.
\end{proof}

Observe that \eqref{eq:splitter-strategy} and
\ref{item:splitter-strategy-iii} imply that for $i>j\ge 1$
we have
$H^{(i)}\subseteq H^{(j)}[N_r^{H^{(j)}}(v^{(j)})]$,
and as $v^{(i)}\in V(H^{(i)})$ by \ref{item:splitter-strategy-i}, there is a path
$P^{(i)}_j$ of length
at most
$r$ from $v^{(j)}$ to $v^{(i)}$ in
$H^{(j)}$. Thus the desired paths in \ref{item:splitter-strategy-ii}
always exist.  

Furthermore, note that by \ref{item:splitter-strategy-ii} it holds that
\begin{equation}\label{eq:splitter2}
\sum_{i=1}^{k}|W^{(i)}| \ \ \le \ \  \sum_{i=2}^k (1{+}r)(i{-}1) \ \ \le \ \
(r{+}1)\sum_{i=1}^{k-1}i \ \ \le \ \ (r{+}1)k^2.
\end{equation}
and let
\begin{equation}
\label{eq:ell(C,r)}
\ell \ \ \coloneqq \ \ (r{+}1)t^2.
\end{equation}
Note that if $H^{(i)}=G^{(i-1)}$ for all $i\in[k]$, then  $\pi$ is a
partial play of the $(\ell,r)$-splitter game. Moreover, condition
\ref{item:splitter-strategy-ii} describes a winning strategy for
Splitter, because if she chooses her sets $W^{(i)}$ according to
condition \ref{item:splitter-strategy-ii} then she wins no matter
which vertices Connector chooses.  
But the lemma is stronger; it even gives Splitter a winning strategy for a generalised game where in each round Connector is allowed to restrict the game to a subgraph $H^{(i)}$ of $G^{(i-1)}$.

\begin{remark}\label{rem:compute-splitter-strategy}
  Following \cite[Remark~4.7]{GroheKS17}, we can compute this winning strategy as follows: for each $i\in[k]$, we compute a breadth-first search tree $T^{(i)}$ for the graph $H^{(i)}$ rooted in $v^{(i)}$ in time $O(\|H^{(i)}\|)$.
  Then for
 each $j\in[1,i{-}1]$
  we use the tree $T^{(j)}$ to compute the path
$P^{(i)}_j$ in time $O(|P^{(i)}_j|)$, which allows us to compute the
set $W^{(i)}$ in time $O(ri)$.
\end{remark}

\subsection{Neighbourhood Covers}

Let $G$ be a graph and $r\in\Nat$. An \emph{$r$-neighbourhood cover} of $G$ is a mapping $\CX\colon V(G)\to 2^{V(G)}$ such that for all $v\in V(G)$, the set $\CX(v)\subseteq V(G)$ is connected in $G$ and $N_r^G(v)\subseteq\CX(v)$. The sets $\CX(v)$ are called the \emph{clusters} of $\CX$. For $X\subseteq V(G)$, we write $X\in\CX$ to denote that $X$ is a cluster of $\CX$, i.e., $X=\CX(v)$ for some $v\in V(G)$.

The \emph{radius} of a non-empty set $X\subseteq V(G)$ is the least $s\in\Nat$ such that there is a $c\in X$, called a \emph{center} of $X$, with $X\subseteq N_s^G(c)$. The \emph{maximum radius} of an $r$-neighbourhood cover $\CX$ is the maximum of the radii of the clusters $X\in\CX$.

The \emph{degree} of a vertex
$v\in V(G)$ in a neighbourhood cover $\CX$ 
is the number of clusters $X\in\CX$ such that $v\in
  X$. The
  \emph{maximum degree} of $\CX$ is the maximum of the degrees of all $v\in V(G)$. 
  Note that
\begin{equation}\label{eq:Delta}
 \sum_{X\in \CX}|X| 
 \ \ \le\ \
 |V(G)|\cdot\Delta,
\end{equation}
where $\Delta$ denotes the maximum degree of $\CX$.

All these notions naturally generalise
from graphs to
$\sigma$-structures by considering the Gaifman graph $G_\CA$ of a
$\sigma$-structure $\CA$.

\begin{lemma}[\cite{GroheKS17}]\label{lem:nb-covers}
  Let~$\CC$ be a nowhere dense class of graphs. Then for
  every $\epsilon>0$ and $r\in\Nat$ there is an algorithm that, given
  a graph $G\in\CC$, computes
  in time~$O(n^{1+\epsilon})$, for $n=|V(G)|$, 
  an $r$-neighbourhood cover $\CX$ of $G$ 
  of radius at most $2r$ and maximum degree
  $O(n^\epsilon)$ and a centre function $c\colon\CX\to V(G)$ such that $X\subseteq N_{2r}^G(c(X))$ for each $X\in\CX$.
\end{lemma}

\subsection{Distance Independent Sets}

Let $d\in\PNat$. Furthermore, let $\CA$ be a $\sigma$-structure. A \emph{$d$-independent set} in $\CA$ is a set $Y\subseteq A$ such that $\dist^\CA(y,y')>d$ for all distinct $y,y'\in Y$.

\begin{lemma}[\cite{GroheKS17}]\label{lem:dist-is}
  Let $\CC$ be a nowhere dense class of $\sigma$-structures.
Then for all $d,m\in\Nat$ and $\epsilon>0$ there is an algorithm that,
given a structure $\CA\in\CC$ and a set $X\subseteq A$, decides
in time $O(n^{1+\epsilon})$, for $n=|\CA|$,
if $\CA$ has a $d$-independent set $Y\subseteq X$ of order $|Y|\ge m$.
\end{lemma}

\subsection{The Removal Lemma}
The main algorithm proceeds by following a strategy for Splitter in
the splitter game. This means that we repeatedly remove elements of a
structure. We have to translate formulas from the original structure
to the structure obtained by removing an element, and this is the
content of the \emph{removal lemma}. The lemma is closely related to
\cite[Lemma~8.2]{GroheKS17} and
\cite[Lemma~7.5]{GroheS18}
but needs to be adapted to our new rank measure.

Recall that by $\vec z_I$ we denote the projection of a tuple $\vec z=(z_1,\ldots,z_k)$
to the coordinates in $I\subseteq[k]$. We extend the notation by
letting $\vec z_{\setminus I}\deff\vec z_{[k]\setminus I}$.

Let $r\in\Nat$.
For every relation symbol $R\in\sigma$  
we let $\tilde{R}_\emptyset\deff R$, 
and for $s\deff\ar(R)$ and for every
set $J\subseteq[s]$ with $J\neq \emptyset$ we introduce
a fresh $(s{-}|J|)$-ary relation symbol $\tilde R_J$.
We let $\tilde\sigma$ be the union of $\sigma$ and the set of all these
new relation symbols,\footnote{In particular, $R_{[s]}$ is a 0-ary relation
  symbol in $\tilde\sigma$.}
and we let $\tilde\sigma_r$ be the extension of
$\tilde\sigma$ by fresh unary relation symbols $S_i$ for all $i\in[r]$.
For every $\sigma$-structure
$\CA$ of order $|A|\ge 2$ and every $c\in A$, we let $\CA\lbag c$ be the
$\tilde\sigma$-structure with universe $A\setminus\{c\}$ and relations
\[
\tilde R_J^{\CA\lbag c}\ \ \coloneqq\ \ 
\big\{\,
  \vec{a}_{\setminus J} \ : \
\vec{a}\in R^{\CA}\text{ and }J=\{j\in[s]\mid a_j=c\}
\,\big\}
\]
for every
$R\in\sigma$ and every
$J\subseteq[\ar(R)]$.
Furthermore, we let $\CA\lbag_r c$ be the
$\tilde\sigma_r$-expansion of $\CA\lbag c$ in which each $S_i$ is
interpreted by the set of all $b\in A\setminus\{c\}$ such that
$\dist^{\CA}(c,b)\le i$.
 Note that (for fixed $\sigma$ and $r$), we can compute
$\CA\lbag_r c$ from $\CA$ and $c$ in time $O(\|\CA\|)$. Also note that
the removal operation ``respects'' the Gaifman graphs in the sense
that
\begin{equation}\label{eq:removalGaifmanGraph}
  G_{\CA\lbag_r c} \  = \ 
  G_\CA[A\setminus\set{c}].
\end{equation}

\begin{lemma}[Removal Lemma]\label{lem:fo-removal}
  Let $p,q\in\Nat$,
  and let $r\deff \delta(p,q)$.
Then for every $\phi(\vec x)\in\FOplus[p,q]$ of vocabulary
$\sigma$ and with $\vec x=(x_1,\ldots,x_k)$ for $k\leq q{-}p$,  and for every set $I\subseteq[k]$ there is a
  formula $\tilde \phi_I(\vec x_{\setminus I})\in\FOplus[p,q]$ of vocabulary $\tilde\sigma_r$ such that for all
  $\sigma$-structures $\CA$ of order $|A|\geq 2$, all $c\in A$, and all $\vec
    a=(a_1,\ldots,a_k)\in A^k$ such that 
    $I=\{i\in[k]\mid a_i=c\}$, we
    have
    \begin{equation}\label{eq:removal}
    \CA\ \models\ \phi[\vec a]
    \quad\iff\quad
    \CA\lbag_r c\ \models\ \tilde\phi_I[\vec a_{\setminus I}].
    \end{equation}
    Furthermore, there is an algorithm that computes $\tilde\phi_I$
    from $\phi(\vec x)$
    and $I$.
\end{lemma}

\begin{proof}
Note that for all distance atoms $\dist(y,z)\leq d'$ that occur in formulas
in $\FOplus[p,q]$ it holds that $d'\leq r$.

We proceed by induction on the construction of $\phi$ and show that if
$\phi\in\FOplus[p',q]$ for some $p'\leq p$, then also $\tilde\phi_I\in\FOplus[p',q]$.

For the base step, we consider atomic $\FO$ formulas of vocabulary $\sigma$ and distance atoms. 
 
If $\phi(x_1,\ldots,x_k)$ is of the form $R(x_{j_1},\ldots,x_{j_s})$
with $s=\ar(R)$ and $j_1,\ldots,j_s\in[k]$, then we let 
  $\tilde\phi_I(\vec x_{\setminus I})\coloneqq \tilde R_J(\vec y)$, where 
$J=\{\nu\in[s] : j_\nu\in I\}$ and
$\vec y$ is the tuple obtained from $(x_{j_1},\ldots,x_{j_s})$ by deleting the entries $x_{j_\nu}$ for all $\nu\in J$.

If $\phi(x_1,\ldots,x_k)$ is of the form $x_{j_1}{=}x_{j_2}$ with $j_1,j_2\in[k]$, then we proceed as follows.
If $j_1,j_2\in I$ then $\tilde{\phi}_I(\vec x_{\setminus I})\coloneqq
\logic{true}$
 (note that in this case $\vec x_{\setminus I}$ does not contain any of the variables $x_{j_1},x_{j_2}$).
If $j_1,j_2\not\in I$ then 
$\tilde{\phi}_I(\vec x_{\setminus I})\coloneqq x_{j_1}{=}x_{j_2}$.
Finally, if $j_1\in I\iff j_2\not\in I$ then $\tilde{\phi}_I(\vec x_{\setminus I})\coloneqq \logic{false}$. 

If $\phi(x_1,\ldots,x_k)$ is of the form $\dist(x_{j_1},x_{j_2})\le d'$ with $j_1,j_2\in[k]$, then we proceed as follows.
If $j_1,j_2\in I$ we let $\tilde\phi_I(\vec x_{\setminus I})\coloneqq\logic{true}$.  If $j_1\in I, j_2\not\in I$, we let 
$\tilde\phi_I(\vec x_{\setminus I})\coloneqq S_{d'}(x_{j_2})$;
and analogously if $j_2\in I, j_1\not\in I$, we let
$\tilde\phi_I(\vec x_{\setminus I})\coloneqq S_{d'}(x_{j_1})$. 
Finally, if $j_1,j_2\not\in I$, we let
\[
  \tilde\phi_I(\vec x_{\setminus I}) \ \ \coloneqq\ \ \dist(x_{j_1},x_{j_2})\le {d'}\ \vee\bigvee_{\substack{1\le d_1,d_2\le d'-1,\\d_1+d_2=d'}}\big(S_{d_1}(x_{j_1})\wedge S_{d_2}(x_{j_2})\big).
\]

This completes the construction of $\tilde\phi_I$ for the induction
base. Note that in each case, equation \eqref{eq:removal} holds;
and if $\phi(\vec x)\in\FOplus[p',q]$ for some $p'\leq p$, then also $\tilde\phi_I\in\FOplus[p',q]$.\medskip

For the induction step, Boolean combinations are handled in the
obvious way: if $\phi=\neg\psi$ then $\tilde\phi_I\deff
\neg\,\tilde\psi_I$, and if $\phi=(\psi\vee\chi)$ then
$\tilde\phi_I\deff (\tilde{\psi}_I\vee\tilde{\chi}_I)$.
Obviously, equation \eqref{eq:removal} holds;
and if $\phi(\vec x)\in\FOplus[p',q]$ for some $p'\leq p$, then also $\tilde\phi_I\in\FOplus[p',q]$.

Now, all that remains to be done to finish the induction step is to
consider formulas $\phi$ that start with an existential quantifier. We
distinguish between two cases.

\emph{Case 1:} 
$\phi(\vec x)$ is of the form $\exists
x_{k+1}\,\big(\dist(x_i,x_{k+1})\le d'\,\wedge\,\psi(\vec x')\big)$, where
$i\in[k]$ and $\vec x'=(x_1,\ldots,x_k,x_{k+1})$.
If $i\in I$, we let
   \[ 
    \tilde\phi_I(\vec x_{\setminus I})\ \ \coloneqq \ \ 
    \tilde\psi_{I\cup\{k+1\}}(\vec x'_{\setminus (I\cup\{k+1\})})
    \ \vee \ 
   \exists x_{k+1}\big(S_{d'}(x_{k+1})\wedge\tilde\psi_{I}(\vec x'_{\setminus I})\big).
 \]
 If $i\not\in I$, we let
    \begin{align*}
    \tilde\phi_I(\vec x_{\setminus I})\ \ \coloneqq \ \ &
    \hphantom{\vee \ \ }\big(\, S_{d'}(x_i)\ \wedge \ \tilde\psi_{I\cup\{k+1\}}(\vec x'_{\setminus (I\cup\{k+1\})})\,\big)\\
   & \vee\ \exists x_{k+1}\,\big(\dist(x_i,x_{k+1})\le d'\ \wedge \ \tilde\psi_{I}(\vec x'_{\setminus I})\big)\\
    &\vee\bigvee_{\substack{1\le d_1,d_2\le
      d'-1\\d_1+d_2=d'}}\Big(S_{d_1}(x_i)\ \wedge\ \exists
      x_{k+1}\,\big(S_{d_2}(x_{k+1})\ \wedge\ \tilde\psi_{I}(\vec
      x'_{\setminus I})\big)\Big). 
  \end{align*}
In both cases it can easily be verified that equation
\eqref{eq:removal} holds; and if $\phi(\vec x)\in\FOplus[p',q]$
for some $p'\leq p$, then also $\tilde\phi_I\in
\FOplus[p',q]$.

\emph{Case 2:}
Case~1 does not apply, and $\phi(\vec x)$ is of the form $\exists
x_{k+1}\,\psi(\vec x')$, with $\vec x'=(x_1,\ldots,x_k,x_{k+1})$.
We let
\[
    \tilde\phi_I(\vec x_{\setminus I})\ \ \coloneqq \ \ 
    \tilde\psi_{I\cup\{k+1\}}(\vec x'_{\setminus(I\cup\{k+1\})})\ \vee \ \exists x_{k+1} \, \tilde\psi_{I}(\vec x'_{\setminus I}).
\]
It is easy to see that equation \eqref{eq:removal} holds;  and if
$\phi(\vec x)\in\FOplus[p',q]$ for some $p'\leq p$,
then also $\tilde \phi_I\in\FOplus[p',q]$.
This completes the proof of Lemma~\ref{lem:fo-removal}.
\end{proof}

We will only apply the Removal Lemma 
to formulas with one free
variable. We state this special case explicitly. 

\begin{corollary}
  Let $p,q\in\Nat$ such that $p<q$
  and let $r\deff \delta(p,q)$.
Then for every $\phi(x)\in\FOplus[p,q]$ of vocabulary
$\sigma$ there are a sentence $\tilde\phi_1\in\FOplus[p,q]$ and a
formula $\tilde\phi_0(x)\in\FOplus[p,q]$ of vocabulary
$\tilde\sigma_r$ such that for all 
  $\sigma$-structures $\CA$ of order $|A|\geq 2$ and all $a,c\in A$ we
  have
  \begin{equation}\label{eq:removal1}
    \CA\ \models\ \phi[a] \ \ \iff \ \
  \begin{cases}
    \ \CA\lbag_r c\ \models\ \tilde\phi_1 & \text{if }a=c,\\
    \ \CA\lbag_r c\ \models\ \tilde\phi_0[a] & \text{if }a\neq c.
  \end{cases} 
\end{equation}
    Furthermore, there is an algorithm that computes $\tilde\phi_1$ and $\tilde\phi_0(x)$
    from $\phi(x)$.
\end{corollary}

We also need an iterated version of this corollary. Let $\sigma$ be a vocabulary and $r\in\Nat$. We let $\tilde\sigma^0_r\coloneqq\sigma$, and for $s\in\Nat$ we let $\tilde\sigma^{s+1}_r\coloneqq\tilde\tau_r$ for $\tau\coloneqq\tilde\sigma^s_r$. For a $\sigma$-structure $\CA$ of order $|A|>s$ and distinct elements $c_1,\ldots,c_s\in A$, we let
\[
\CA\lbag_r c_1\cdots c_s \ \ \coloneqq \ \ (\cdots((\CA\lbag_r c_1)\lbag_r c_2)\cdots)\lbag_r c_s.
\]
Then $\CA\lbag_r c_1\cdots c_s$ is a $\tilde\sigma^s_r$-structure with
universe $A\setminus\{c_1,\ldots,c_s\}$ and Gaifman graph
$G_{\CA}[A\setminus\{c_1,\ldots,c_s\}]$
(this follows inductively by using
equation~\eqref{eq:removalGaifmanGraph}). Furthermore, we have
\[
  \sigma \ = \ \tilde\sigma^0_r
  \ \subseteq\ \tilde\sigma^1_r
  \ \subseteq\ \tilde\sigma^2_r
  \ \subseteq\ \cdots\ .
\]

\begin{corollary}\label{cor:iterated-removal}
  Let $p,q\in\Nat$ such that $p<q$
  and let $r\deff \delta(p,q)$. Moreover, let $s\in\PNat$.
Then for every $\phi(x)\in\FOplus[p,q]$ of vocabulary
$\sigma$ there are sentences $\tilde\phi_1,\ldots,\tilde\phi_s\in\FOplus[p,q]$ and a formula $\tilde\phi_0(x)\in\FOplus[p,q]$ of vocabulary $\tilde\sigma^s_r$ such that for all
  $\sigma$-structures $\CA$ of order $|A|\geq s{+}1$, all distinct $c_1,\ldots,c_s\in A$, and all $a\in A$ we
  have
  \begin{equation}\label{eq:removal2}
    \CA\ \models\ \phi[a]\iff
  \begin{cases}
    \CA\lbag_r c_1\cdots c_s\ \models\ \tilde\phi_1 & \text{if }a=c_1,\\
    \hspace{1cm}\vdots & \hspace{5mm}\vdots\\
    \CA\lbag_r c_1\cdots c_s\ \models\ \tilde\phi_s & \text{if }a=c_s,\\
    \CA\lbag_r c_1\cdots c_s\ \models\ \tilde\phi_0[a] & \text{if }a\in A\setminus\{c_1,\ldots,c_s\}.
  \end{cases} 
\end{equation}
    Furthermore, there is an algorithm that computes $\tilde\phi_1,\ldots,\tilde\phi_s$ and $\tilde\phi_0(x)$
    from $\phi(x)$.
\end{corollary}

\subsection{The Main Algorithm}

Theorem~\ref{theo:GKS} immediately follows from the following lemma. For a formula $\phi(x)$ and a structure $\CA$, it will be convenient to let
\[
\CS(\CA,\phi) \ \ \coloneqq \ \ \big\{ a\in A\bigmid \CA\models\phi[a]\big\}.
\]

\begin{lemma}\label{lem:GKS}
  Let $\CC$ be a nowhere dense class of $\sigma$-structures.
  Furthermore, let $\phi(x)$ be an $\FOplus$ formula. Then for every $\epsilon>0$ there is an algorithm that, given $\CA\in\CC$, computes the set $\CS(\CA,\phi)$ in time $O(|\CA|^{1+\epsilon})$.
\end{lemma}

\begin{proof}
  Let $\CC_G$ be the class of all Gaifman graphs of structures in $\CC$ and all their subgraphs. Then $\CC_G$ is nowhere dense.

  We assume
  without loss of generality that $\free(\phi)=\set{x}$ (in
  case that $\free(\phi)=\emptyset$, we replace $\phi$ with the
  formula $(\phi \wedge x{=}x)$).
  
  Let $\phi(x)$ be an $\FOplus$ formula
  with $|\free(\phi)|=1$, and let
  $\rankletter\deff\rk(\phi)$. Then, $\phi\in\FOplus[\rankletter{-}1,\rankletter]$.  Let
  \begin{equation*}
     r \ \ \coloneqq \ \ \myrho(\rankletter{-}1,\rankletter).
  \end{equation*} 
  We choose  $t\coloneqq t(\CC_G,2r)$ according to
  Lemma~\ref{lem:splitter-strategy} and let
  $\ell\coloneqq (2r{+}1)t^2$
  as in \eqref{eq:ell(C,r)}
  (for $2r$ instead of $r$). 

  Let $\epsilon>0$. Without loss of generality we assume
  $\epsilon\le\frac{1}{2}$ (which
  implies that $\epsilon^2\le \epsilon/2$).
  We let
  \begin{equation}\label{eq:eps}
    \epsilon' \ \coloneqq \ \frac{\epsilon}{2t}.
  \end{equation}
  Furthermore, we choose a
  $c_1\in\PNat$ such that for all
  graphs
  $G\in\CC_G$  
  of order $n\coloneqq|G|>c_1$ the $r$-neighbourhood cover of
  Lemma~\ref{lem:nb-covers} has maximum degree at most
  $n^{\epsilon'}$ (we can
  apply the lemma with $\epsilon=\epsilon'/2$ to achieve this). Moreover, we let $c_2\coloneqq c_1+\ell$.

  Suppose the input structure is $\CA\in\CC$ and let
  $n\coloneqq|A|$. If $n\le c_1$,
  we compute the set $\CS(\CA,\phi)$
  directly by brute force. Suppose that $n>c_1$, and let $\CX$ be an
  $r$-neighbourhood cover of $G_{\CA}$ of radius at most $2r$ and with
  maximum degree at most $n^{\epsilon'}$. 

  By Theorem~\ref{thm:rankGaifman}, $\phi$ is equivalent to a Boolean combination $\phi'(x)$ of 
  $r$-local $\FOplus$ formulas $\psi(x)$ with $\rk(\psi)\le\rankletter$ and
  basic local sentences 
  \begin{equation}\label{eq:blGKS}
  \xi \ \ \coloneqq \ \ \exists x_1\cdots\exists x_s\,\Bigl(\bigwedge_{1\le i<j\le s}\dist(x_i,x_j)>2r'\;\wedge\;\bigwedge_{1\le i\le s}\psi(x_i)\Bigr),
  \end{equation}
  where $r'\le r$, $s\le\rankletter$, and $\psi(x)$ is an $r'$-local $\FOplus$ formula with $\rk(\psi)\le\rankletter$. Since $r'\le r$, the formula $\psi(x)$ is also $r$-local.

  For each $r$-local formula $\psi(x)$ that appears in $\phi'(x)$,
  either directly in the Boolean combination or as part of a basic
  local sentence $\xi$, we shall compute the set $\CS(\CA,\psi)$. 

  Then we can evaluate the basic local sentences $\xi$ of the form
  \eqref{eq:blGKS} as follows: we note that $\CA\models\xi$ if and
  only if $\CA$
  has a $2r'$-independent set $Y\subseteq\CS(\CA,\psi)$ of order $|Y|\ge s$. We can use Lemma~\ref{lem:dist-is} to decide if this is the case in time $O(n^{1+\epsilon'})$. Once we have computed the truth values of all basic local sentences $\xi$ in $\phi'$ and the sets $\CS(\CA,\psi)$ for all local formulas $\psi(x)$ in $\phi'$, 
  we can easily evaluate the Boolean combination $\phi'$ and compute the set $\CS(\CA,\phi')=\CS(\CA,\phi)$.
  Thus it remains to compute $\CS(\CA,{\psi})$ for all formulas $\psi(x)$ in a finite set $\Psi$ of $r$-local formulas of rank at most $\rankletter$. 

  We shall recursively compute the set $\CS\big(\CA[X],\psi\big)$ for every $X\in\CX$ and every $\psi(x)\in\Psi$. Observe that if $N_r^\CA(a)\subseteq X$ then $a\in \CS(\CA,\psi)\iff a\in \CS(\CA[X],\psi)$, because $\psi(x)$ is $r$-local. Thus
  \begin{equation}\label{eq:evaluate-root}
  \CS(\CA,\psi) \ \ = \ \ \big\{\;a\in A\bigmid a\in \CS\big(\CA[\CX(a)],\psi\big)\;\big\}.
  \end{equation}
 
  Let us move on to the recursive step.
  It will be convenient to think of our recursive algorithm in terms of the recursion tree it builds. 
  Each node of the recursion tree will be indexed by a tuple $\vec X=(X_1,\ldots,X_k)$, where $k\in[0,t]$ and the $X_i$ are subsets of $A$. The root is indexed by the empty tuple $()$. In the following, we will no longer distinguish between the nodes and their indices and refer to a tuple $\vec X$ as a node of the tree. 

  For each node $\vec X$ that is not a leaf of the tree we define a
  vocabulary $\sigma_{\vec X}$, which will always be of the form
  $\tilde\sigma^i_r$ for some $i\ge 0$, a $\sigma_{\vec X}$-structure
  $\CA_{\vec X}$, its Gaifman graph $G_{\vec X}$, an $r$-neighbourhood
  cover $\CX_{\vec X}$ of $G_{\vec X}$ of radius at most $2r$ and
  maximum degree at most $|G_{\vec X}|^{\epsilon'}$, a centre map $c_{\vec X}$ for
  $\CX_{\vec X}$, and a finite set $\Psi_{\vec X}$ of $r$-local
  formulas with one free variable,
  of vocabulary $\sigma_{\vec X}$ and of rank at most
  $\rankletter$. For the root $()$, we let
  $\sigma_{()}\coloneqq\sigma=\sigma^0_r$, $\CA_{()}\coloneqq\CA$,
  $G_{()}\coloneqq G_{\CA}$, $\CX_{()}\coloneqq\CX$, and
  $\Psi_{()}\coloneqq\Psi$.   
  
  For nodes $\vec X\neq()$ that are not leaves, in addition to the
  structure $\CA_{\vec X}$, the graph $G_{\vec X}$, the neighbourhood
  cover $\CX_{\vec X}$ and the set $\Psi_{\vec X}$ of formulas,  we
  will define a graph $H_{\vec X}$, a vertex $v_{\vec X}$, and a set
  $W_{\vec X}$. They will be defined in such a way that for each node
  $\vec X=(X_1,\ldots,X_k)$ the vertices $v^{(i)}\coloneqq
  v_{(X_1,\ldots,X_i)}$, sets $W^{(i)}\coloneqq W_{(X_1,\ldots,X_i)}$,
  graphs $G^{(0)}\coloneqq G_{()}$
  and $G^{(i)}\coloneqq
  G_{(X_1,\ldots,X_i)}$, $H^{(i)}\coloneqq H_{(X_1,\ldots,X_i)}$, for
  $i\in[k]$, satisfy conditions \eqref{eq:splitter-strategy} and
  \ref{item:splitter-strategy-i}--\ref{item:splitter-strategy-iii} of
  Lemma~\ref{lem:splitter-strategy} with $2r$ instead of
  $r$. Furthermore, we define a breadth-first search tree $T_{\vec X}$
  of the graph $H_{\vec X}$ rooted in $v_{\vec X}$. 

  Consider node $\vec X=(X_1,\ldots,X_k)$. If $k\ge 1$ and $|X_k|\le c_2$, then $\vec X$ is a leaf of the tree.
  If $k=0$ or $|X_k|>c_2$, the node $\vec X$ has a child $\vec XY\coloneqq(X_1,\ldots,X_k,Y)$ for every $Y\in\CX_{\vec X}$. 
  
  Let $\vec X=(X_1,\ldots,X_{k})$ and $\vec X'=(X_1,\ldots,X_k,Y)$ for some $Y\in\CX_{\vec X}$. Our goal at the node $\vec X'$ is to compute $\CS(\CA_{\vec X}[Y],\psi)$ for every $\psi(x)\in\Psi_{\vec X}$. From these sets $\CS(\CA_{\vec X}[Y],\psi)$ we can compute 
  \begin{equation}\label{eq:eval-from-children}
  \CS(\CA_{\vec X},\psi) \ \ = \ \ \Big\{\;a\in A_{\vec X}\Bigmid a\in \CS\big(\CA_{\vec X}[\CX_{\vec X}(a)],\psi\big)\;\Big\},
  \end{equation}
  exploiting that $\psi(x)$ is $r$-local and that $\CX_{\vec X}$ is an
  $r$-neighbourhood cover of the Gaifman graph $G_{\vec X}$ of
  $\CA_{\vec X}$.
  For the root $\vec X=()$ of the recursion tree,
  this enables us to
  compute $\CS(\CA,\phi)$, as argued above.

So consider a node $\vec Y=\vec XY=(X_1,\ldots,X_k,Y)$ for some $k\in\Nat$. If $|Y|\le c_2$, then $\vec Y$ is a leaf, and we compute $\CS(\CA_{\vec X}[Y],\psi)$ for every $\psi(x)\in\Psi_{\vec X}$ by brute force.
Suppose that $|Y|>c_2$. For $i\in[0,k]$, let $\vec X^{(i)}\coloneqq(X_1,\ldots,X_{i})$ and $G^{(i)}\coloneqq G_{\vec X^{(i)}}$. Furthermore, if $k\ge 1$, let $v^{(i)}\coloneqq v_{\vec X^{(i)}}$, $W^{(i)}\coloneqq W_{\vec X^{(i)}}$, and $H^{(i)}\coloneqq H_{\vec X^{(i)}}$. 
  Observe that $Y$ is an element of the neighbourhood cover $\CX_{\vec
    X^{(k)}}$ of $G^{(k)}$ and hence $Y\subseteq V(G^{(k)})$. We let
  $v_{\vec Y}\coloneqq v^{(k+1)}\coloneqq c_{\vec X^{(k)}}(Y)$ and
  $H_{\vec Y}\coloneqq H^{(k+1)}\coloneqq G^{(k)}[Y]$. Then we choose
  $W_{\vec Y}\coloneqq W^{(k+1)}\subseteq Y$ according to
  Lemma~\ref{lem:splitter-strategy}\ref{item:splitter-strategy-ii}
  with $2r$ instead of $r$, and we let $G_{\vec Y}\coloneqq
  G^{(k+1)}\coloneqq H^{(k+1)}[N_{2r}^{H^{(k+1)}}(v^{(k+1)})\setminus
  W^{(k+1)}]$ (according to
  Lemma~\ref{lem:splitter-strategy}\ref{item:splitter-strategy-iii}).
  Observe
  that actually we have
  $G^{(k+1)}= H^{(k+1)}[Y\setminus
  W^{(k+1)}]$, because the radius of the neighbourhood cover
  $\CX_{\vec X^{(k)}}$ is at most $2r$ and
  $v^{(k+1)}=c_{\vec X^{(k)}}(Y)$, and hence
  $Y\subseteq N_{2r}^{G^{(k)}}(v^{(k+1)})$.
  \medskip
  
  If $W^{(k+1)}=\emptyset$, we
  proceed as follows.
  We let $\sigma_{\vec Y}\coloneqq\sigma_{\vec X}$ and $\CA_{\vec
    Y}\coloneqq\CA_{\vec X}[Y]$. We let $G_{\vec Y}\coloneqq
  G_{\CA_{\vec Y}}$ and compute an $r$-neighbourhood cover $\CX_{\vec
    Y}$ of $G_{\vec Y}$
  of radius at most $2r$ and maximum degree at most $|Y|^{\epsilon'}$,
  and a centre map $c_{\vec Y}$ using
  Lemma~\ref{lem:nb-covers}. Moreover, we let $\Psi_{\vec
    Y}\coloneqq\Psi_{\vec X}$. 
  The children of $\vec Y$ are $\vec YZ$ for $Z\in\CX_{\vec Y}$, and
  for all $\psi(x)\in\Psi_{\vec Y}=\Psi_{\vec X}$ we can compute
  $\CS(\CA_{\vec Y},\psi)=\CS(\CA_{\vec X}[Y],\psi)$ from the
  recursively computed sets $\CS(\CA_{\vec Y}[Z],\psi)$ for
  $Z\in\CX_{\vec Y}$ as in
  \eqref{eq:eval-from-children}.
  \medskip
  
  Now suppose that $W^{(k+1)}=\{w_1,\ldots,w_s\}\neq\emptyset$. To
  define $\sigma_{\vec Y}$, suppose that $\sigma_{\vec
    X}=\tilde\sigma^i_r$. Then we let $\sigma_{\vec
    Y}\coloneqq\tilde\sigma^{i+s}_r$ and $\CA_{\vec
    Y}\coloneqq\CA_{\vec X}[Y]\lbag_r w_1\cdots w_s$. We let $G_{\vec
    Y}\coloneqq G_{\CA_{\vec Y}}$. Note that $V(G_{\vec
    Y})=Y\setminus\{w_1,\ldots,w_s\}$, and as $|Y|>c_2=c_1+\ell$ and
  $s=|W^{(k+1)}|\le \ell$ we have $|G_{\vec Y}|>c_1$. We compute an
  $r$-neighbourhood cover $\CX_{\vec Y}$ of $G_{\vec Y}$
  of radius at most $2r$ and maximum degree at most $|Y\setminus\set{w_1,\ldots,w_s}|^{\epsilon'}$
  and a centre
  map $c_{\vec Y}$ using Lemma~\ref{lem:nb-covers}. The children of
  $\vec Y$ are $\vec YZ$ for $Z\in\CX_{\vec Y}$.  

  For each $\psi(x)\in\Psi_{\vec
    X}\subseteq\FOplus[\rankletter{-}1,\rankletter]$
  we compute
  sentences $\tilde\psi_1,\ldots,\tilde\psi_s\in\FOplus[\rankletter{-}1,\rankletter]$ and a
  formula $\tilde\psi_0(x)\in\FOplus[\rankletter{-}1,\rankletter]$ by
  Corollary~\ref{cor:iterated-removal}. The corollary shows that we
  can compute $\CS(\CA_{\vec X}[Y],\psi)$ from the truth values of the
  sentences $\tilde\psi_i$ in $\CA_{\vec Y}$ and the set
  $\CS(\CA_{\vec Y},\tilde\psi_0)$. We apply
  Theorem~\ref{thm:rankGaifman} to each of the $\tilde\psi_i$ and
  obtain a formula $\tilde\psi_i'$ in Gaifman normal form
(for each $i\in[0,s]$).
  As explained above for $\phi(x)$, from this formula in Gaifman normal
  form we obtain a set $\Psi_i$ of $r$-local formulas
  with one free variable and
  of rank at most
  $\rankletter$, such that we can reduce the evaluation of
  $\tilde\psi_i$ in $\CA_{\vec Y}$ to the evaluation of all formulas
  in $\Psi_i$ in $\CA_{\vec Y}$. We let
  $\Psi_\psi\coloneqq\bigcup_{i=0}^s\Psi_i$. Then we can reduce the
  evaluation of $\psi$ in $\CA_{\vec X}[Y]$ to the evaluation of all
  formulas in $\Psi_\psi$ in $\CA_{\vec Y}$. Finally, we let
  $\Psi_{\vec Y}\coloneqq\bigcup_{\psi\in\Psi_{\vec
      X}}\Psi_\psi$. Then we can compute $\CS(\CA_{\vec X}[Y],\psi)$
  for all $\psi\in\Psi_{\vec X}$ from the recursively computed
  $\CS(\CA_{\vec Y}[Z],\chi)$ for $Z\in\CX_{\vec Y}$ and
  $\chi(x)\in\Psi_{\vec Y}$. 
  \medskip
  
  This completes the description of the recursive algorithm. It remains to analyse the algorithm.
  We first observe that the depth of the recursion is at most $t$; this follows from Lemma~\ref{lem:splitter-strategy}. 
  Moreover, for each node $\vec X=(X_1,\ldots,X_k)$ the sizes of the
  sets $W_{\vec X_i}$ for $\vec X_i\coloneqq(X_1,\ldots,X_i)$ add up
  to at most $\ell$ by \eqref{eq:splitter2}
  (for $2r$ instead of $r$)
  and our choice of $\ell$. 
  Whenever we remove the elements of a non-empty set $W_{\vec X_i}$ of
  size $s$ we increase the vocabulary from $\sigma_{\vec
    X_{i-1}}=\tilde\sigma^j_r$ to $\sigma_{\vec
    X_i}=\tilde\sigma^{j+s}_r$. Thus $\sigma_{\vec
    X}=\tilde\sigma^{j'}_r$ for $j'=\sum_{i=1}^k|W_{\vec X_i}|$. As
  $j'\le\ell$, it follows that $\sigma_{\vec
    X}\subseteq\tilde\sigma^\ell_r\eqqcolon\sigma^*$. Note that the
  size of $\sigma^*$ only depends the input formula $\phi(x)$ (via
  $\sigma$ and $r$) and the class $\CC$ (via
  $t=t(\CC_G,2r)$), but not on
  the size of the input structure.  
  
  Thus,
  all formulas $\psi(x)\in\Psi_{\vec X}$ are of a vocabulary that is
  contained in $\sigma^*$.
  Furthermore, we may assume that all formulas in
  $\Psi_{\vec X}$ are normalised (that is, we normalise each formula
  immediately after computing it). As there is only a finite number of
  normalised formulas $\psi(x)$ of vocabulary $\sigma^*$ and of rank
  at most $\rankletter$, there is an a constant upper bound (only
  depending on $\phi$ and $\CC$) on the size of the sets $\Psi_{\vec
    X}$ and the length of the formulas in these sets. Thus the number
  and size of formulas we need to evaluate at each node only
  contributes a constant factor to the running time. 
  
  All structures $\CA_{\vec X}$ have a vocabulary contained in $\sigma^*$, which means that we can bound their size $\|\CA_{\vec X}\|$ and the size of their Gaifman graph $G_{\vec X}$ by $O(|\CA_{\vec X}|^{1+\epsilon'})$. Also note that for $\vec X=(X_1,\ldots,X_k)$ we have 
  \[
  |\CA_{\vec X}| \ \ \le \ \ |X_k| \ \ =: \ \ n_{\vec X}.
  \]
  At each non-leaf node, we have to compute the neighbourhood cover $\CX_{\vec X}$ and center map $c_{\vec X}$, which can be done in time $O(n_{\vec X}^{1+\epsilon'})$ by Lemma~\ref{lem:nb-covers}, and the breadth-first search tree $T_{\vec X}$ of $H_{\vec X}$, which can also be done in time $O(\|H_{\vec X}\|)=O(n_{\vec X}^{1+\epsilon'})$. By using the trees $T_{\vec X_i}$ of the ancestors $\vec X_i$ of $\vec X$ in the tree, we can compute $W_{\vec X}$ in constant time (cf.~Remark~\ref{rem:compute-splitter-strategy}). 
  
  Then we can compute $\CA_{\vec X}$ in time $O(\|\CA_{\vec X}\|)=O(n_{\vec X}^{1+\epsilon'})$. Finally, assuming $\vec X=\vec X'X_k$, we have to compute $\CS(\CA_{\vec X'}[X_k],\psi)$ for each $\psi(x)\in\Psi_{\vec X'}$ from the recursively computed $\CS(\CA_{\vec X}[Y],\chi)$ for $Y\in\CX_{\vec X}$ and $\chi(x)\in\Psi_{\vec X}$. 
  Combining the sets $\CS(\CA_{\vec X}[Y],\chi)$ to $\CS(\CA_{\vec
    X},\chi)$ requires time $O(\sum_{Y\in\CX_{\vec X}}|Y|)=O(n_{\vec
    X}^{1+\epsilon'})$ by \eqref{eq:Delta} since the maximum degree of
  $\CX_{\vec X}$ is at most
  $n_{\vec X}^{\epsilon'}$. 
  Computing $\CS(\CA_{\vec X'}[X_k],\psi)$ for all
  $\psi(x)\in\Psi_{\vec X'}$ from the
  sets $\CS(\CA_{\vec X},\chi)$ for
  $\chi(x)\in\Psi_{\vec X}$ involves a large constant factor to handle
  all the formulas, and may involve computing $2r'$-independent sets
  using Lemma~\ref{lem:dist-is} to evaluate basic local sentences, but
  is possible in time $O(n_{\vec X}^{1+\epsilon'})$. 

  Overall, the running time at each non-leaf node $\vec X$ is $O(n_{\vec X}^{1+\epsilon'})$ (not including the time for the recursive calls). The running time at a leaf node is constant.
  
  We obtain the following recurrence for the running time $T(j,m)$ the algorithm spends in the subtree rooted at a node $\vec X=(X_1,\ldots,X_{t-j})$ with $n_{\vec X}=m$:
  \begin{align*}
    T(0,m)\ &=\ O(1) &\text{ for all }m,\\
    T(j,m) \ &= \ O(1) &\text{for all $j\in[t]$ and all $m\le c_2$},\\
    T(j,m) \ &= \ \sum_{Y} T(j{-}1,m_Y)+O(m^{1+\epsilon'})&\text{for all $j\in[t]$ and all $m>c_2$}.
  \end{align*}
  The sum ranges over all $Y$ in the neighbourhood cover $\CX_{\vec
    X}$, and
  $m_Y\coloneqq n_{\vec X Y}$.
  Recall that
  $\sum_{Y}m_Y\le m^{1+\epsilon'}$, by \eqref{eq:Delta} using that the
  maximum degree of $\CX_{\vec X}$ is at most
  $m^{\epsilon'}$.
  Note that the clause $T(0,m)=O(1)$
  is justified by the fact that the height of the tree is at most
  $t$. Therefore, all nodes
  of the form $(X_1,\ldots,X_t)$ must be leaves.
  
  The same recurrence was obtained in \cite{GroheKS17}. The
  straightforward analysis,
  carried out in \cite{GroheKS17} (and using, among other things, the
  inequality $\sum_i y_i^p \leq (\sum_i y_i)^p$ for $y_i\geq 0$ and
  $p\geq 1$, based on the fact that the $\ell_p$ norm of a vector
  $(y_i)_i$ is bounded from above by its $\ell_1$ norm), yields
  $T(j,n)=O(n^{1+2j\epsilon'})=O(n^{1+\epsilon})$
  for all $j\in[0,t]$. In particular, this
  bounds the running time at the root 
  to $O(n^{1+\epsilon})$
  and completes the proof of
  Lemma~\ref{lem:GKS}. 
\end{proof}

\section{Conclusions}\label{sec:conclusion}
We introduce a new rank measure for formulas of first-order logic and prove a rank-preserving version of Gaifman's Theorem.

We also prove that Gaifman's Theorem does not preserve the quantifier
rank, or the \emph{outer quantifier rank}, to be precise. However, it
remains unclear how much the (inner, outer, total) quantifier rank
needs to increase when transforming a formula into Gaifman normal
form. Keisler and Lotfallah's proof of Gaifman's Theorem
\cite{KeislerLotfallah2004} yields a linear upper bound on the total
quantifier rank, which implies linear upper bounds for the inner
quantifier rank and the outer quantifier rank as well. It is an
interesting open problem to determine tight upper and lower bounds for
the different quantifier ranks and on their interplay with the width
and the radius.

\paragraph{Acknowledgements}
We thank Charlotte Lenz for bringing to our attention that Section~7
of \cite{GroheS18} contains a serious flaw in the proof of
\cite[Theorem~7.1]{GroheS18}.  Our effort in fixing this flaw led to
the results presented in this paper.
We thank Steffen van Bergerem for his feedback on an earlier version
of this paper.

\printbibliography

\end{document}